\documentclass[aps, pra, 10pt, superscriptaddress, twocolumn, longbibliography, floatfix]{revtex4-1}

\usepackage[nointlimits]{amsmath}
\usepackage{amsthm}
\usepackage{graphicx,nicefrac}
\usepackage{subfigure}
\usepackage{color}
\usepackage{txfonts}
\usepackage{mathtools}
\usepackage{hyperref}
\usepackage{tabularx}
\hypersetup{colorlinks=true, citecolor=blue, linkcolor=blue}
 
\usepackage{physics}

\usepackage{amsfonts,amsmath,amssymb}
\usepackage{array,bm,color}
\usepackage{epsfig,graphicx,nomencl,revsymb4-1,upgreek,url}
\usepackage{hyperref}
\usepackage{algpseudocode}
\usepackage{xspace}
\usepackage{tabularx}
\usepackage[normalem]{ulem}
\usepackage{changes}
\usepackage{enumitem}
\usepackage{booktabs}
\usepackage{mathtools}

\usepackage{tikz}
\usepackage{quantikz}

\usepackage{algorithm}
\usepackage{algpseudocode}

\makeatletter
\DeclareFontFamily{OMX}{MnSymbolE}{}
\DeclareSymbolFont{MnLargeSymbols}{OMX}{MnSymbolE}{m}{n}
\SetSymbolFont{MnLargeSymbols}{bold}{OMX}{MnSymbolE}{b}{n}
\DeclareFontShape{OMX}{MnSymbolE}{m}{n}{
    <-6>  MnSymbolE5
   <6-7>  MnSymbolE6
   <7-8>  MnSymbolE7
   <8-9>  MnSymbolE8
   <9-10> MnSymbolE9
  <10-12> MnSymbolE10
  <12->   MnSymbolE12
}{}
\DeclareFontShape{OMX}{MnSymbolE}{b}{n}{
    <-6>  MnSymbolE-Bold5
   <6-7>  MnSymbolE-Bold6
   <7-8>  MnSymbolE-Bold7
   <8-9>  MnSymbolE-Bold8
   <9-10> MnSymbolE-Bold9
  <10-12> MnSymbolE-Bold10
  <12->   MnSymbolE-Bold12
}{}

\let\llangle\@undefined
\let\rrangle\@undefined
\DeclareMathDelimiter{\llangle}{\mathopen}%
                     {MnLargeSymbols}{'164}{MnLargeSymbols}{'164}
\DeclareMathDelimiter{\rrangle}{\mathclose}%
                     {MnLargeSymbols}{'171}{MnLargeSymbols}{'171}
\makeatother

\usepackage{tikz}
\usetikzlibrary{arrows}

\def\bra#1{{\langle #1 |}}
\def\ket#1{{| #1 \rangle}}

\newcommand{\eg}[1]{#1}

\newcommand{\pauli}{\mathbb{P}}

\newcommand{\toevent}{\Delta}
\newcommand{\eegev}{\mathcal{D}}
\newtheorem{lemma}{Lemma}
\newtheorem{definition}{Definition}

\newtheorem{theorem}[lemma]{Theorem}

\tikzset{
    vertex/.style={circle,draw,minimum size=1.5em},
    edge/.style={->,> = latex'}
}

\begin{document}
 %TC:ignore
 \title{Simulating Quantum Error Correction beyond Pauli Stochastic Errors}
\author{Jordan Hines}
\thanks{jhhines@sandia.gov}
\affiliation{Quantum Performance Laboratory, Sandia National Laboratories, Albuquerque, NM 87185}
\author{Corey Ostrove}
\affiliation{Quantum Performance Laboratory, Sandia National Laboratories, Albuquerque, NM 87185}
\author{Kenneth Rudinger}
\affiliation{Quantum Performance Laboratory, Sandia National Laboratories, Albuquerque, NM 87185}
\author{Stefan Seritan}
\affiliation{Quantum Performance Laboratory, Sandia National Laboratories, Livermore, CA 94550}
\author{Kevin Young}
\affiliation{Quantum Performance Laboratory, Sandia National Laboratories, Livermore, CA 94550}
\author{Robin Blume-Kohout}
\affiliation{Quantum Performance Laboratory, Sandia National Laboratories, Albuquerque, NM 87185}
\author{Timothy Proctor}
\affiliation{Quantum Performance Laboratory, Sandia National Laboratories, Livermore, CA 94550}

\begin{abstract} 
Quantum error correction (QEC), the lynchpin of fault-tolerant quantum computing (FTQC), is designed and validated against well-behaved \textit{Pauli stochastic} error models.  But in real-world deployment, QEC protocols encounter a vast array of other errors---coherent and non-Pauli errors---whose impacts on quantum circuits are vastly different than those of stochastic Pauli errors. The impacts of these errors on QEC and FTQC protocols have been largely unpredictable to date due to exponential classical simulation cost. 
Here, we show how to accurately and efficiently model the effects of coherent and non-Pauli errors on FTQC, and we study the effects of such errors on syndrome extraction for surface and bivariate bicycle codes, and on magic state cultivation. Our analysis suggests that coherent error can shift fault-tolerance thresholds, increase the space-time cost of magic state cultivation, and can increase logical error rates by an order of magnitude compared to equivalent stochastic errors.  
These analyses are enabled by a new technique for mapping any Markovian circuit-level error model with sufficiently small error rates onto a detector error model (DEM) for an FTQC  circuit. The resulting DEM enables Monte Carlo estimation of logical error rates and noise-adapted decoding, and its parameters can be analytically related to the underlying physical noise parameters to enable approximate strong simulation. 
\end{abstract}
 %TC:endignore
\maketitle
\section*{Introduction}

Achieving the computational speedups that quantum algorithms promise for classically-intractable problems \cite{gidney2025factor, baczewski2024stopping, troyer2017reaction, low2025simulations} will require extraordinarily high-fidelity logic operations \cite{campbell2017faulttolerant, fowler2012surface, Proctor2025} achievable only in fault-tolerant quantum computing (FTQC) enabled by quantum error correction (QEC) \cite{shor, steane1996, knill2000errorcorrection, aharonov1997faulttolerant}, which has now been demonstrated at the smallest scales on a growing variety of quantum processors \cite{Acharya2025, Bluvstein2026, rosenfeld2025cultivation, ryananderson2024teleportation, lacroix2024color, SalesRodriguez2025}.  Fault-tolerant logic operations on encoded logical qubits are carefully designed to catch and mitigate the effects of combinations of errors that are \textit{expected} to occur during the operation (typically, errors on a small number of qubits), so that logical operations have higher fidelity than physical operations. However, QEC and fault-tolerant operations are designed and assessed based on their performance against \textit{simple} errors---specifically \textit{Pauli stochastic} errors---whose impacts on QEC can be simulated efficiently. In contrast, errors on real-world qubits and control operations  are \emph{not} Pauli stochastic \cite{Mądzik2022, Blume-Kohout2017, Carignan_Dugas_2024, rudinger2025heisenberg, gross2024coherent} (unless they are forced into Pauli-stochastic form at the cost of control complexity \cite{wallman2016tailoring}). How well QEC and fault-tolerant operations perform against realistic circuit-level errors, especially \textit{coherent} errors produced by miscalibrated Hamiltonians, remains poorly understood, despite results indicating that non-Pauli-stochastic errors can impact QEC very differently from standard Pauli stochastic models \cite{miller2025simulation, leblond2025logical, myers2025simulating, takou2025coherent, tuloup2026computing}, because simulating arbitrary unitary evolution on many qubits is famously hard.

Here, we remove this barrier and determine the effects of general circuit-level error models on key FTQC primitives including syndrome extraction (for surface codes and bivariate bicycle codes) and magic state cultivation.  We find that the impact of coherent errors varies dramatically with context.  In the worst case, coherent errors add up constructively within syndrome extraction to cause logical failure rates more than $8\times$ higher than equivalent stochastic Pauli errors and shift code thresholds. In other contexts, such as surface code decoding, the coherence of errors appears to have little impact on decoder performance, although decoders customized to coherent errors using our simulation and modeling tools do achieve better logical error rates than decoders based on Pauli-twirled models in the low-distance, high-infidelity regime. We find that certain coherent two-qubit gate errors increase the spacetime cost of magic state cultivation \cite{gidney2024magic, rosenfeld2025cultivation} relative to stochastic Pauli error of analogous strength, demonstrating that FTQC resource requirements depend not just on the strength of errors, but on their precise form.

\begin{figure*}[t!]
    \centering
    \includegraphics[width=\linewidth]{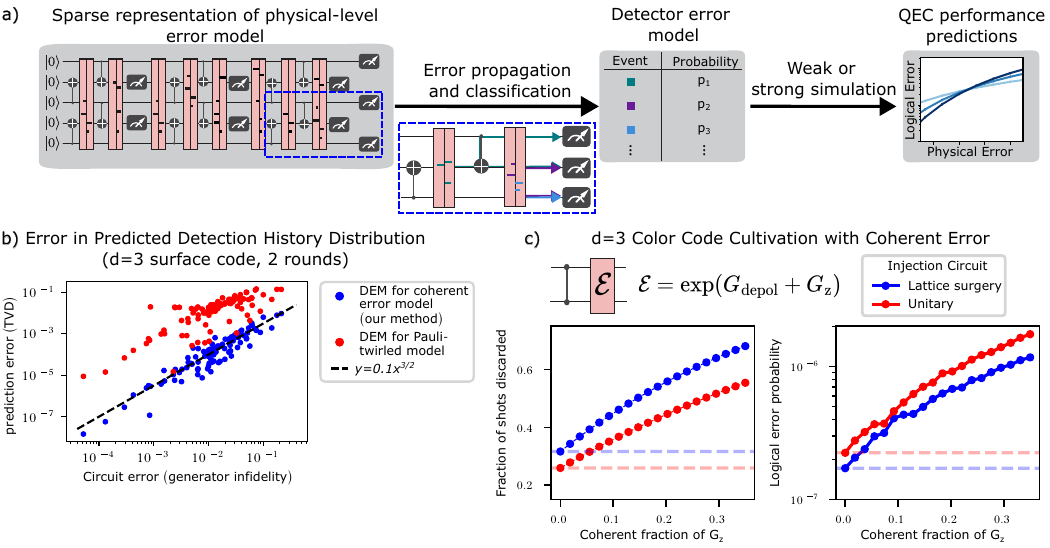}
    \caption{\textbf{Efficient prediction of quantum error correction circuits with non-Pauli errors.} (a) We use an efficient method for propagating errors through Clifford circuits to determine how non-Pauli errors impact a QEC circuit. This allows us to (approximately) map errors onto events in a detector error model, which can be efficiently sampled from to predict logical performance. (b) Our method outperforms DEMs generated from Pauli-twirled error models in predicting detection history distributions by roughly 100x (quantified by TVD between the predicted and true distributions). (c) We simulated the injection of an $S$ state (as a proxy for a $T$ state) followed by a double logical $H$ check in a $d=3$ color code with error models with the same gate infidelities, but different relative contributions of stochastic and coherent errors. We include depolarizing error plus additional error on each gate, and the dominant source of non-depolarizing error in these simulations was $ZZ$ error on $CZ$ gates. As the contribution of coherent error to the gate error increases, the discard rate and logical error probability both increase relative to a Pauli stochastic model (dashed lines).}
    \label{fig:fig1_new}
\end{figure*}

These analyses are enabled by a new, highly-flexible simulation toolchain that produces \emph{detector error models} (DEMs) \cite{blumekohout2025detector, arms2025estimating, derks2024designing, gidney2021stim} for FTQC gadgets experiencing general Markovian errors, including coherent errors.  DEMs enable rapid simulation of QEC and noise-informed decoding, both for logical memory \cite{Terhal_2015} and for key elements of FTQC like logical gates or magic state production \cite{gidney2024cultivation, bravyi2005distillation}. They compress circuit-level error models into a compact probabilistic model of QEC syndrome measurement and logical observable bits, which models exactly (and only) the effects needed to simulate and assess QEC performance. DEMs can be efficiently \textit{inferred} (sometimes with large sampling requirements) from experiments involving coherent errors \cite{blumekohout2025detector, takou2025estimating, arms2025estimating}, but existing methods for \textit{constructing} DEMs from physical-level models are limited to Pauli stochastic error models \cite{gidney2021stim, derks2024designing} or brute-force simulation \cite{takou2025coherent}, which severely limits their ability to describe realistic errors' impact on large-scale QEC.  

To break this barrier and study the effects of general errors on QEC, we built on recent work in perturbative simulation of error generator models \cite{miller2025simulation, blumekohout2022taxonomy} to create an efficient perturbative algorithm for computing DEMs that model arbitrary small Markovian errors in non-adaptive circuits. The algorithm is computationally efficient and scalable to high-distance QEC codes for realistic error models. This enables accurate, interpretable modeling and simulation of QEC with hardware-tailored error models. 

\section*{Results}
\subsection*{Efficient Construction of Models for QEC Circuits}
Not every parameter in a physical-level error model is needed to predict the outcome distribution of QEC circuits. The key insight we use in our approach is that most parameters in these models can be removed systematically by considering the \emph{perturbative} impacts of errors on measurement results. We compute these impacts and use them to construct a \emph{detector error model (DEM)}, which models only the outcomes of \emph{detectors}, Pauli observable measurements that have a definite outcome $(\pm 1)$ in the absence of error. A DEM is a probabilistic model consisting of independent \emph{DEM events}, each of which is set of detectors that flip simultaneously with a specified probability. 

Existing approaches to constructing DEMs build them from Pauli stochastic models for physical-qubit operations. Our method starts with errors on physical qubits represented as fully general completely positive trace-preserving (CPTP) maps, covering a much larger range of errors. We represent such errors as $\mathcal{E}=\exp(G)$, where the error's \emph{generator} $G$ is represented in the \emph{elementary error generator (EEG)} basis. The EEG basis has two properties critical to our method. First, errors in quantum processors are typically described by generators that are \emph{sparse} in the EEG basis---i.e., the number of EEGs needed to represent the generator is much less than $16^n$ for a system of $n$ physical qubits---making errors efficient to represent. Second, EEGs can be efficiently propagated through Clifford circuits  \cite{miller2025simulation} with mid-circuit measurements (see Methods). These properties allow us to efficiently propagate the errors in our model [Fig~\ref{fig:fig1_new}(a)] like Pauli errors. We then approximately combine all errors in the circuit into a single \emph{circuit error channel} $\mathcal{E}_c=\exp(G_{c})$, which captures the effects of accumulation of error (e.g.,  coherent addition/cancellation of errors). 

The circuit's error generator can be decomposed as a linear combination of components, each of which causes a \emph{single} DEM event, i.e., $G_{c} = G_1 + G_2 + \cdots + G_k$ where the effects of $\exp(G_i)$ on the circuit are exactly modeled by a DEM with a single event $D_i$. Identifying these components is efficient in the EEG basis---it is done by sorting EEGs into \emph{DEM event classes}, which are determined by commutation relations between an EEG's Pauli indices and the circuit's detectors (see Methods). We can then approximate $\mathcal{E}_c$ as a product of channels with the generators $G_1, G_2, \dots, G_k$ using the Zassenhaus expansion \cite{Wang_2019} to $j$th order: 
\begin{align}
    \mathcal{E}_{c} \approx \exp(G_1)\exp(G_2)\cdots \exp(G_k)\prod_{i=2}^j W_i, \label{eq:event_zassenhaus2}
\end{align}
where $W_i$ denotes the $i$th order Zassenhaus term, and each $W_i$ for $i>1$ can also be decomposed into components only impacting a single DEM event. We then estimate the rate of each DEM event by estimating detector flip probabilities from the channels $\exp(G_i)$ (see Methods).

\subsection*{Surface Code Syndrome Extraction with Coherent Error}
QEC is typically simulated using Pauli error models---frequently constructed via Pauli twirling \cite{silva2008correctable}---as proxies for more general error models.
We studied the consequences of using Pauli-twirled models to predict syndrome and logical observable measurements using our method, which, in contrast to Pauli-twirled models, accounts for coherent accumulation of error. We simulated $2$ rounds of $d=3$ rotated surface code syndrome extraction with sparse coherent error models. 
This example is small enough that a statevector simulation can compute the true probabilities of detection histories. We found that our method predicts the probability distribution over syndromes and logical observables much more accurately than Pauli-twirled models in the presence of coherent errors [Fig.~\ref{fig:fig1_new}(b)]. The total variation distance (TVD) between the distribution predicted by our method and the true distribution is typically 1-2 orders of magnitude smaller than the TVD between the distribution predicted from the Pauli-twirled model and the true distribution.

At leading order, our DEMs typically estimate the probabilities of the dominant detection histories to within $1\%$ relative precision, and many probabilities are estimated to within $0.1\%$ relative precision (see Appendix~\ref{app:sims}). The error in our DEMs' predictions of the probability distribution (as measured by TVD) scales as $\epsilon_{\textrm{gen}}^{1.5}= O(\epsilon_{h_P}^3)$ [Fig.~\ref{fig:fig1_new}], where $\epsilon_{\textrm{gen}}$ is the \emph{generator fidelity} of the circuit error, which is defined as 
\begin{equation}
\label{eq:epsilon_ham}
    \epsilon_{\textrm{gen}} = \sum_{P \in \mathbb{G}_{H}} \epsilon_{h_P}^2,
\end{equation}  
where the circuit error generator is
\begin{equation}
\label{eq:eeg_ham}
    G_{c} = \sum_{P \in \mathbb{G}_H} \epsilon_{h_P} H_P,
\end{equation}  
where $\mathbb{G}_{H}$ is the set of Hamiltonian EEGs in the circuit error generator. Note that Eqs.~\ref{eq:epsilon_ham} and~\ref{eq:eeg_ham} are specific to purely coherent error models. In Appendix~\ref{app:sims}, we further validate the accuracy of our DEMs for error models with more general sparse Markovian error models.

\subsection*{Determining the Scaling of Logical Memory Error Rates}

The Difference in detection probabilities between coherent and stochastic error models suggests that these physical error models may lead to different logical qubit performance.
We the impact of coherent error and other non-Pauli errors on logical memory error rates in simulations of rotated surface code syndrome extraction, using minimum-weight perfect matching (MWPM) decoding. Across only $10$ randomly-sampled sparse CPTP error models (see Appendix~\ref{app:sims} for details), we saw a variation of over $0.002$ in the threshold CNOT infidelity [Fig.~\ref{fig:surface_thresholds}(a)]. We obtained these results for families of error models in which all stochastic, active, and Pauli-correlation error rates are uniformly scaled by a factor of $p$ and all Hamiltonian rates are uniformly scaled by a factor of $\sqrt{p}$. We also found that the threshold can vary drastically across models that have different amounts of coherent error, but map to the same Pauli-twirled model. Fig.~\ref{fig:surface_thresholds}(b) shows the logical error probability of surface codes memory circuits for distances $d=3,5,7$ for two error models where each gate has fixed generator infidelity,
\begin{equation}
\label{eq:epsilon_gen}
    \epsilon_{\textrm{gen}} = \sum_{P \in \mathbb{G}_{H}} h_P^2 + \sum_{P \in \mathbb{G}_{S}} s_P,
\end{equation}  
where $\mathbb{G}_{H}$ and $\mathbb{G}_S$ are the set of Hamiltonian and stochastic Pauli EEGs in the models, respectively, but the models have different relative contributions of coherent and stochastic error to the generator infidelity (see Appendix~\ref{app:sims} for details). We found that models with the same CNOT generator infidelity, but higher contribution from coherent error have a higher logical error rate, particularly in the in the near-threshold regime.
The threshold occurs at a CNOT infidelity of approximately $0.006$ in the error model with a large contribution to generator infidelity from coherent error ($75\%$ of the CNOT generator infidelity being from coherent errors) and a CNOT infidelity of approximately $0.012$ in a purely stochastic version of the error model---indicating that the specifics of the error is important for quantifying how far below threshold a system is.  

Logical error rates are decoder-dependent, and the decoder itself can be tailored to the errors on hardware using a DEM. We explore the impact of using an error-tailored DEM on surface code decoding by comparing MWPM decoders with two different weightings----one given by a DEM generated via our method, and the other given by a DEM produced from a Pauli-twirled error model for the gates [Fig~\ref{fig:surface_thresholds}(d)]. The DEM generated via our method results in slightly lower logical error rates for high CNOT infidelities (particularly those above threshold) and $d=3$, but it led to no significant logical error improvement over decoding based on the Pauli-twirled DEM at lower CNOT infidelities or higher code distances---suggesting that differences in DEM event rates captured by our DEMs may not lead to significant changes in surface code decoding in the high-distance, sub-threshold regime required for large-scale fault-tolerant quantum computations.

\begin{figure}
    \centering
    \includegraphics{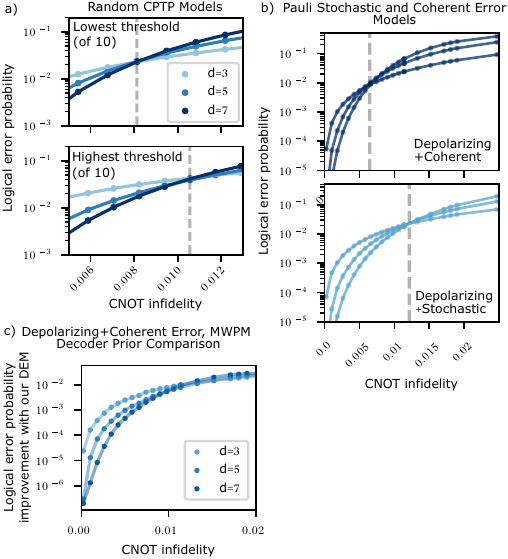}
    \caption{\textbf{Surface Code logical performance with coherent and Pauli stochastic error.} (a) We simulated surface code syndrome extraction with 10 random sparse CPTP error models. Here, we show the results from the models with the lowest and highest thresholds. (b) We simulated models with strong intrinsic ($ZX$, $IX$, $ZI$) CNOT error, but with different amounts of stochastic ($S_P$) and coherent ($H_P$) error. Error models with larger contributions from coherent error exhibit a lower threshold CNOT infidelity than the purely stochastic error model.  (c) We compared two decoder priors for MWPM in our simulation with coherent error plus depolarizing error---priors from a DEM derived using our method and priors from a DEM derived from the Pauli-twirled version of the error model. We see a significant decrease in logical error probability with the DEM derived from our method at CNOT infidelities above threshold, but this difference rapidly drops off below threshold.}
    \label{fig:surface_thresholds}
\end{figure}

The above threshold studies utilize single-parameter families of error models, where the strength of all errors is determined by a single scaling parameter. However, sources of physical-level error do not all change uniformly in practice, and can often be varied independently or in small groups.
We explore the interaction of independent changes in two coherent error parameters on different physical-level gates in syndrome extraction for the  $[[144,12,12]]$ gross code \cite{bravyi2024memory} with 2 rounds of syndrome extraction and the Beam decoder, a belief-propagation-based decoder \cite{ye2025beam}. In our error model, we include coherent $H_X$ errors on every idling qubit during mid-circuit measurements, and we include $H_{XX}$ error after every CNOT gate. Fig~\ref{fig:bb} shows the scaling of the logical error rate with these two coherent error parameters. 
Coherent errors in this error model can cancel within syndrome extraction, resulting in fewer detection events [Fig~\ref{fig:bb}(b)] than predicted by the Pauli-twirled model---we observe that when the coherent idle error rate is $h_x \approx 0.0032$, a lower logical error rate is achieved when the CNOT $H_{XX}$ EEG rate is a small, negative value $h_{xx} \approx -0.013$ than when there is no CNOT coherent error---and in this parameter regime, the Pauli-twirled model \emph{overestimates} the logical error rate (for the values we simulated, we observe underestimates up to $24\%$). We also find that using a decoder with priors determined by the Pauli-twirled model does not significantly impact logical error rates for these error models (see Appendix~\ref{app:sims}).

\begin{figure}
    \centering
    \includegraphics{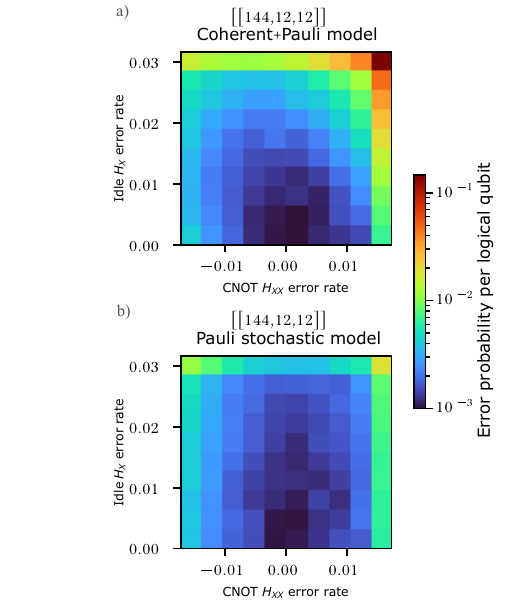}
    \caption{\textbf{Scaling of gross code logical error probability with coherent error.} We simulate the [[144,12,12]] gross code (decoding using the Beam Decoder \cite{ye2025beam}), changing the strengths of two coherent errors. (a) The logical error rate per qubit is minimized at a nonzero $H_{XX}$ error rate for the CNOT gates, with the location of the minimum being dependent on the rate of the idle $H_X$ error. (b) In contrast, in a Pauli-twirled analogue of the model, the $H_{XX}$ error rate minimizing the logical error rate does not shift with the $H_X$ idle error rate. In the regime where physical-level errors cancel, the Pauli-twirled model overestimates the logical error rate.}
    \label{fig:bb}
\end{figure}

\subsection*{Magic State Cultivation with Coherent Error}

\begin{figure*}
    \centering
    \includegraphics{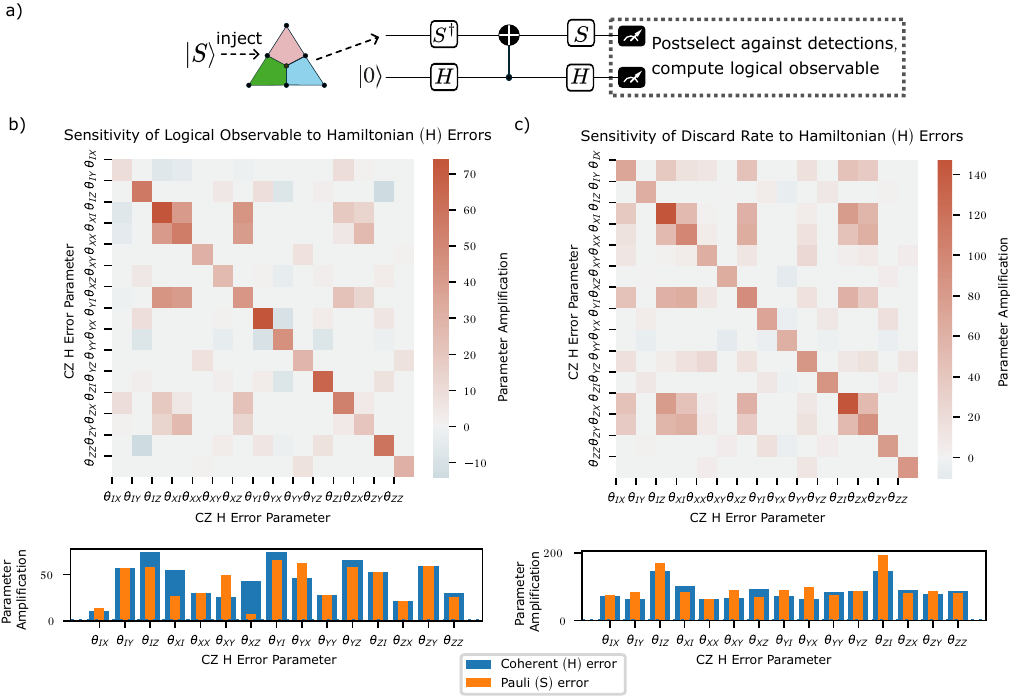}
    \caption{\textbf{Magic State Cultivation with Coherent Error.} (a) High-level diagram of the cultivation protocol we simulated. We simulated the injection of an $S$ state (as a proxy for a $T$ state) followed by a double logical $H$ check in a $d=3$ color code. We model each $CZ$ gate as having a fixed infidelity of 0.00165, and we vary the relative contributions of stochastic and coherent $ZI$, $IZ$, and $ZZ$ error on CZ gates and Z error on single-qubit gates gates. (b) The sensitivity of the rate of postselection to coherent error at leading order is quadratic in the $H$ error parameters described by its sensitivity matrix. (c) The end-of-circuit logical observable is is particularly sensitive to $H_{IY}$ and $H_{YI}$ errors. Furthermore, it is more sensitive to $H_{YI}$ and $H_{ZZ}$ error than their $S$ analogues---suggesting that these $H$ errors lead to a higher logical error probability.}
    \label{fig:cultiv}
\end{figure*}

We analyzed the impact of coherent errors on magic state cultivation \cite{gidney2024cultivation, itogawa2025zero, Chamberland2020}, an efficient method for generating high-fidelity $T$ states, and found that that coherent errors can be significantly more detrimental than Pauli stochastic errors to the initial stages of cultivation. Specifically, we considered cultivation without escape into a high-distance code (i.e., injection followed by a double-check of the logical observable) [Fig.~\ref{fig:cultiv}a] using two different methods of state injection---unitary state injection \cite{itogawa2025zero} and lattice-surgery-based injection \cite{Chamberland2020}. 
Due to the substantial added complexity of simulating $T$ gates, here we analyze circuits $S$ state cultivation, a Clifford proxy for $T$ state cultivation with the same circuit structure. 
We simulated cultivation with depolarizing error on both single- and two-qubit gates, plus additional Z-type ($IZ$, $ZI$, and $ZZ$) errors (with total generator infidelity of 0.0012) on the two-qubit gates ($CZ$) and additional $Z$ errors on the single-qubit gates (with total generator infidelity of 0.00024) [Fig.~\ref{fig:fig1_new}c], and we varied the relative contribution of coherent and stochastic error to the $IZ$, $ZI$ and $ZZ$ errors, while keeping the infidelity of the $CZ$ gates approximately fixed (at $0.000165$). The fraction of shots kept after postselection decreases for circuits using both styles of state injection as the rate of coherent error increases [Fig~\ref{fig:fig1_new}c], indicating detrimental coherent error amplification throughout the circuit. Similarly, the logical error probability increases by approximately a factor of $10$ between purely stochastic error and $35\%$ of the additional Z-type error being coherent error. Additionally, the two injection circuits we studied each exhibit different benefits. Unitary injection results in a lower discard rate but higher logical error probability than the lattice-surgery-based circuit. Additionally, the difference in discard fraction between the two circuits becomes larger with higher amounts of coherent error.

To better understand the sensitivity of cultivation to CZ gate coherent errors, we computed explicit functional relationships between the coherent error rates of the gates in our physical-level models and the detector Pauli observables used in our DEMs. At leading order, $\langle P \rangle = 1 - 2\bar{p}_{P}$, where $\bar{p}_{P}$ denotes the probability of a DEM event that flips detector $P$. Let $\theta$ be a vector of all of the parameters in our error generator model---i.e., the coherent error rates for each EEG used in the error model for each physical-qubit gate ($H$ and $CZ$).  At leading order, the detector expectation value is quadratic in the physical coherent error rates, i.e., $\langle P \rangle  = 1 + \boldsymbol{\theta}^T S_P \boldsymbol{\theta}$ (see Methods). $S_P$ is a \emph{sensitivity matrix} encoding how $\langle P \rangle$ depends on the coherent error parameters. Restricting our attention to cultivation with unitary injection, in Figure~\ref{fig:cultiv}(b), we show this sensitivity matrix for the logical observable in a $d=3$ color code cultivation circuit with no escape stage. The logical observable is especially sensitive to the $H_{XZ}$ and $H_{IZ}$ errors on the CZ gate, it has relatively low sensitivity to $H_{IX}$ and $H_{ZX}$ errors, and it has low sensitivity to cross-terms involving distinct $H$-type EEGs. We compare the sensitivity of our circuit to an error generator term $\epsilon_{H_P} H_P$ to the approximate Pauli-twirled equivalent $\theta^2_{CZ, H_P} S_P$, where  $\theta_{CZ, H_P}$ denotes the rate of the $H_P$ error on the $CZ$ gate [Figure~\ref{fig:cultiv}(b), bottom] and see that the detector is $15\%$ more sensitive to $H_{ZZ}$ error on $CZ$ gates than $S_{ZZ}$ error on $CZ$ gates, and it is $28\%$ more sensitive to $H_{IZ}$ error than $S_{IZ}$ error---suggesting that it contributes to the increase in logical error probability we observed. 

We also compute the relationship between the rate of postselection in cultivation and physical coherent error parameters [Fig.~\ref{fig:cultiv}(c)], finding that it is $4.8\%$ more slightly sensitive to coherent $ZZ$ error than stochastic Pauli $ZZ$ error.
In contrast, this detector is undersensitive to coherent $IZ$ and $ZI$ errors relative to the analogous Pauli stochastic errors (by $14\%$ and $25\%$, respectively)---which counteracts the amplification of coherent $ZZ$ errors, but is a subdominant effect in our simulations since these errors are significantly smaller in out models (see Appendix~\ref{app:sims}). Our analysis yields insights about errors not included in our simulations as well---$H_{XI}$ and $H_{XZ}$ are more strongly amplified relative to stochastic errors in our circuits than $H_{ZZ}$, but their absolute level of impact of the circuit is much smaller. Similarly, $Y$-type Hamiltonian errors ($YI$, $IY$, and $YY$) all have a lower impact than their Pauli stochastic analogues on the rate of postselection---though, interestingly, the logical observable is more sensitive to Hamiltonian $YI$ error than Pauli stochastic $YI$ error.

\section*{Discussion}

Precise predictions of FTQC performance require compact, scalable models that accurately capture the dynamics that emerge from errors on the physical gates. We have shown how to construct such models for physical gates with arbitrary small Markovian errors, and demonstrated that non-Pauli-stochastic errors can significantly change logical memory error rates, fault-tolerance thresholds, and the performance of magic state cultivation. Our approach enables accurate, hardware-specific performance predictions for fault-tolerant circuits using error models from detailed physical-level characterization \cite{Nielsen2020-rd, Blume-Kohout2017}. Furthermore, we are able to efficiently compute approximate functional relationships between errors event probabilities and physical-gate error parameters, providing a new, analytic approach for systematic tailoring of QEC to device-specific errors. We anticipate that these methods will provide an avenue for low-cost improvements to QEC that leverage the results of low-level physics modeling and physical-qubit characterization. 

We have only explored the impacts of a few types of error, leaving open the question of how other hardware non-idealities---including crosstalk error, spatially inhomogeneous error, and drift---impact FTQC primitives. Our method is readily capable of answering this question. We also anticipate our approach being extended to simulate more complex FTQC primitives than what we have studied here---in particular, protocols with active feedback---enabling assessment of many protocols for performing non-Clifford gates as well as real-time calibration. Furthremore, our method can be extended to include full simulation of physical non-Clifford gates, leveraging existing work reducing the cost of simulating circuits with few $T$ gates \cite{wan2026simulating} to mitigate the increased computational cost, addressing the question of how these gates change the implications of physical-level errors on FTQC relative to Clifford gates. 

Our results show that DEMs are capable of capturing the dominant (leading-order) dynamics in FTQC circuits, but they also point to limitations of DEMs caused by their inability to capture anticorrelations (unless negative event probabilities are allowed). Errors outside of our method's framework, such as leakage errors, can also cause anticorrelations, and are often highly detrimental to QEC \cite{Miao2023}. It is an interesting open question how these interact with other errors in FTQC and how to efficiently model the resulting dynamics.

\section*{Methods}
\subsection*{Modeling and Propagating Non-Pauli Errors in QEC Circuits}

We describe a circuit as a list of instructions to apply a series of operations, $C = L_kL_{k-1}\dots L_1$, where each $L_i$ is a \emph{layer} of instructions executed in parallel. In this work, we consider only circuits consisting of Clifford gates, stabilizer state preparation, and computational basis measurements, including mid-circuit measurements (MCMs). For simplicity, we assume that all MCMs include post-measurement reset of the state in the state $\ket{0}$. Our method also applies to circuits with non-Clifford operations, at the expense of increased computational cost.

We model the error in a circuit layer with no MCMs with a post-gate error channel, i.e., we represent an imperfect implementation of layer $L_i$ as $\mathcal{E}\mathcal{U}(L_i)$, where $\mathcal{U}(L_i)$ is the superoperator representation of the error-free implementation of $L_i$. For layers with MCMs, we model the error with both a post-layer and a pre-layer error channel. More general mid-circuit measurement noise can be introduced with the use of a gadget consisting of additional virtual qubits and a CNOT gate \cite{wysocki2026detailed}. We express each error channel as $\mathcal{E} = \exp(G)$, and we call $G$ an error \emph{generator}. We express error generators in the \emph{elementary error generator} (EEG) basis \cite{blumekohout2022taxonomy}.  Each EEG is associated with a particular error \textit{sector}--Hamiltonian (H), Pauli-stochastic (S), Pauli-correlation (C), or Active (A)--and is indexed by one or two $n$-qubit Pauli operators (see Appendix~\ref{app:theory}) that help encode its action. 
If we express an error channel $\mathcal{E}$ as $\mathcal{E}=\exp(G)$, then $G$ may be decomposed as
\begin{equation}
    G = \sum_{P} \epsilon_{h_P} H_P + \sum_{P} \epsilon_{S_P} S_P + \sum_{P, Q} \epsilon_{A_{P,Q}} A_{P,Q} + \sum_{P,Q} \epsilon_{C_{P,Q}} C_{P,Q} 
\end{equation}
We expect that in a system with H error rates of $\epsilon_{h_P} = O(\epsilon)$ the rates of S, C, and A error generators are of $O(\epsilon^2)$, so the contributions from C and S terms that arise at second order from H error generators should be considered along with first order S, C, and A terms \cite{blumekohout2022taxonomy}. Furthermore, S, C, and A EEGs can contribute to changes in detector statistics at first order in their error rates, whereas the leading-order contribution from H EEGs appears at $O(\epsilon_{h_P}^2)$ (see Appendix~\ref{app:theory}). We therefore consider the contributions of $H$ terms to $O(h_P^2)$ when doing leading-order approximations in this work.

To propagate errors through a circuit with mid-circuit measurements, we construct the \emph{expanded circuit}, a circuit that reproduces the measurement outcomes of the original circuit but without mid-circuit measurements. We construct this circuit by deferring measurements with the help of of auxiliary qubits (see Figure~\ref{fig:fig1}b). 

Absorbing pre-MCM error channels into the previous circuit layer's error (for notational simplicity), the imperfect implementation of a Clifford circuit $C$ is
\begin{equation}
\Lambda(C)=\mathcal{E}_{L_d}\mathcal{U}(L_d)\cdots \mathcal{E}_{L_0}\mathcal{U}(L_0).
\end{equation}
Propagating all errors to the end of the circuit, 
\begin{equation}
    \Lambda(C) = \mathcal{E}'_k\mathcal{E}'_{k-1}\dots \mathcal{E}'_1\mathcal{U}(C) = \mathcal{E}_c\mathcal{U}(C),
\end{equation}
where $\mathcal{U}(C)$ is the unitary superoperator for the perfect implementation of $C$, and we define
\begin{equation}\label{eqn:prod_of_propagated}
    \mathcal{E}_c = \mathcal{E}'_k\mathcal{E}'_{k-1}\dots \mathcal{E}'_1 = \exp(G_{c})
\end{equation}
where
\begin{equation}
    \mathcal{E}'_j = \mathcal{U}_{k}\dots\mathcal{U}_{j+1}\mathcal{E}_j \mathcal{U}_{j+1}^{-1} \dots \mathcal{U}_{k}^{-1}.
\end{equation}
We call $\mathcal{E}_c$ the \emph{circuit error channel} and $G_{c}$ the \emph{circuit error generator}, and they describe the dynamics of \emph{all} errors that occur in the circuit. $G_{c}$ (as well as $\mathcal{E}_c$) is computationally hard to compute exactly, but we can efficiently approximate it using the Baker–Campbell–Hausdorff (BCH) formula \cite{BLANES2004135}, 
\begin{align}
    G_{c} &\approx \sum_{i=1}^k G'_i + \sum_{i=2}^{j} B_i,
\end{align}
where $G'_i$ is the generator of $\mathcal{E}_i'$ and $B_j$ is the $j$th order BCH term. When $\mathcal{E}_1, \mathcal{E}_2, \dots, \mathcal{E}_d$ are sparse in the EEG basis, the $j$th order BCH approximation to $G_{EOC}$ is also sparse.

\begin{figure}[ht!]
    \centering
    \includegraphics[width=\linewidth]{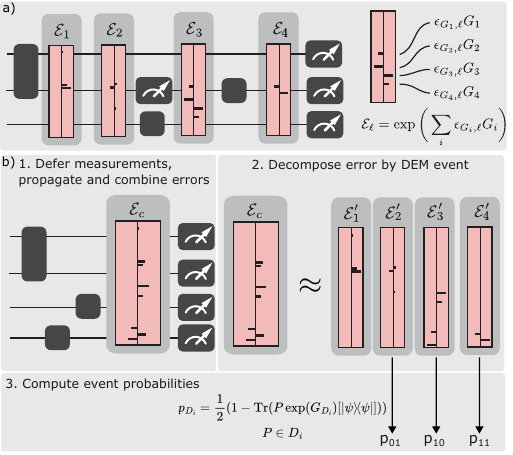}
    \caption{\textbf{Efficient prediction of quantum error correction circuits with non-Pauli errors.} (a) We model errors in a QEC circuit with sparse error generator models, where each error channel is represented as an exponential of a generator that is sparse in the elementary error generator basis. (b) Our method estimates the impact of all physical-level errors on a circuit by first propagating all errors to the end of the circuit, then approximately combining them with a BCH expansion. Then, our method decomposes that circuit error channel approximately into individual error channels that each cause one DEM event, from which the probability of each DEM event can be efficiently estimated.}
    \label{fig:fig1}
\end{figure}

\subsection*{Detector Error Models}
Our method constructs a \emph{detector error model (DEM)}, a coarse-grained model that tracks only the impact of errors on measurement results, not the errors' origin or form.  
A DEM models the outcomes of \emph{detectors}, Pauli observables $P$ that have definite outcome in the absence of errors---i.e., $\bra{\psi}P\ket{\psi} = \pm 1$, where $\ket{\psi}$ is the state produced by the circuit in the absence of errors. For notational simplicity, we use signed Pauli observables, and demand that $\bra{\psi}P\ket{\psi} = 1$ for all detectors $P$. 

For a circuit $C$, we specify a set of detectors $\mathbb{D}$, forming a basis for the set of commuting measurements to be modeled. 
A \emph{detection history} is a record of the results of each detector measurement.
A DEM models the probability distribution of detection histories. It consists of a set of \emph{DEM events} $\{ D_1, D_2, \dots D_k\}$ ,where $D_i \subseteq \mathbb{D}$, where the events have associated probabilities $\{p_1, p_2, \dots p_k\}$. These events are sampled from as binomial random variables to produce detection histories.

Any Pauli fault in a Clifford circuit can be efficiently propagated through the circuit to determine which measurement results it flips, and therefore what DEM event it causes. This is the idea underlying the construction of DEMs from stochastic Pauli error models \footnote{In this discussion, we implicitly assume all Pauli faults in the circuit are independent. In the event that faults are modeled by a higher-rank Pauli channel, any fault can be mapped to a detector bit string, but the possibility of anticorrelated Pauli errors implies that the Pauli channel may not be exactly modeled by a DEM.}. For a Pauli $Q$ and a set of detectors $\mathbb{D}$, we define a function $\toevent$ that maps the Pauli error $Q$ to the DEM event caused by a Pauli $Q$ error immediately before measurement in the expanded circuit---formally, $\toevent(Q) = D$ for $D \subseteq \mathbb{D}$, if for all $P \in \mathbb{D}$, $P \in D$ if and only if $[P,Q] \neq 0$.

\subsection*{Relating Error Generators to DEM Events}
Our approach to creating a DEM from an error generator model [Fig~\ref{fig:fig1}] is inspired by how Pauli stochastic error models are mapped to DEMs.
Like Pauli errors, EEGs, and hence sparse error generators, can be efficiently propagated through Clifford circuits [Fig.~\ref{fig:fig1}b]. However, the errors in our models are \emph{not} independent stochastic faults,
meaning that there is no exact mapping between the circuit error and DEM events. We instead construct an \emph{approximate} mapping of the circuit error to independent events by decomposing it into simpler error channels that each cause a single DEM event. We decompose the circuit error generator as
\begin{equation}\label{eq:decompose_generator}
    G_{c} = G_1 + G_2 + \cdots + G_k, 
\end{equation} 
where each component $G_i$ generates an error channel that causes a \emph{single} DEM event. Identifying these single-event components is efficient for an error generator expressed in the EEG basis, as each EEG impacts only a single DEM event (or no DEM events), as determined by its Pauli indices. We call the set of all EEGs whose Pauli indices all satisfy $\Delta(P)=D$ for a DEM event $D \in \mathcal{P}(\mathbb{D})$ a \emph{DEM event class} of EEGs for event $D$. In Appendix~\ref{app:theory}, we prove Theorem~\ref{thm:eeg_to_event}, which shows that any generator consisting of EEGs in a single DEM event class generates detection histories modeled by a single-event DEM.

\begin{theorem}\label{thm:eeg_to_event}
Let $G$ be an error generator, and suppose its EEG decomposition consists of entirely EEGs of the same DEM event class $D$. Then the circuit error $\mathcal{E}_D = \exp(G)$ is perfectly modeled by the DEM with event $D$ with probability $\frac{1}{2}(1-\Tr(P\mathcal{E}_d[\psi]))$, where $P \in D$.
\end{theorem}

Using the decomposition of $G_{\textrm{EOC}}$ in Eq.~\ref{eq:decompose_generator}, we decompose $\mathcal{E}_c$ using the $j$th order Zassenhaus formula \cite{Wang_2019},
\begin{align}
    \mathcal{E}_{c,j} = \exp(\sum_{G \in \mathbb{G}_1} \epsilon_G G)\exp(\sum_{G \in \mathbb{G}_2} \epsilon_G G)\cdots \exp(\sum_{G \in \mathbb{G}_k} \epsilon_G G)\prod_{i=2}^j W_i, \label{eq:main_event_zassenhaus}
\end{align}
where $\mathbb{G}_1, \mathbb{G}_2, \dots, \mathbb{G}_k \subseteq \mathbb{G}$ are the DEM event classes of EEGs, $W_i$ denotes the $i$th order Zassenhaus term, and we typically take $j=1$. Zassenhaus terms for $i>2$ are \emph{not} single-DEM-event channels---for example, the second-order term is 
\begin{align}
    W_2 & = \exp\left(\frac{1}{2}\left( [G_1, G_2] + [G_1, G_3] + \dots + [G_{k-1}, G_k]\right)\right). 
\end{align} 
However, these higher-order terms can be iteratively decomposed into single-event channels to have their effects incorporated into the DEM.

Our technique uses the single-event channels in Eq.~\eqref{eq:main_event_zassenhaus} in isolation to estimate the probability of each corresponding DEM event. This step imposes a subtle simplifying approximation---it treats each single-event channel as causing an independent stochastic error. Unless the single-event channels are Pauli stochastic channels, errors may combine differently across channels than predicted by a DEM beyond leading order. 
Therefore, higher-order approximations to detection statistics require computing \emph{compositions} of single-event channels, in order to identify events that must be added to the DEM to capture these effects. The set of possible events at order $k$ in EEG rates is efficient to compute (see Appendix~\ref{app:higher_order_events}):

\begin{theorem}\label{thm:higher_order_events}
Let $E_1$ the the set of all DEM events $\Delta(P)$, where P is a Pauli appearing in the indices of the EEGs in the propagated layer error channels $\mathcal{E}'_1,\mathcal{E}'_{2},\dots ,\mathcal{E}'_k$ (Eq.~\ref{eqn:prod_of_propagated}), and let $E_k$ be the set of DEM events that can be expressed as the symmetric difference of up to $k$ DEM events in $E_1$ (i.e., the net effect of $k$ events happening). Then at $k$th order, the DEM describing the circuit's detector histories to order $k$ contains only events in $E_k$. 
\end{theorem}

Theorem~\ref{thm:higher_order_events} does \emph{not} guarantee that the DEM events all have positive rates. Negative DEM event rates can occur when there are \emph{anticorrelations} in detection histories. However, every error model consisting of only Markovian error described by CPTP maps has a leading-order DEM with no negative-rate events (see Appendix~\ref{app:cptp}). Theorem~\ref{thm:higher_order_events} accounts for all sources of approximation in our construction, including treating the single-event channels as independent. It is possible to create a DEM that accounts for non-stochastic addition across single-event channels, but this comes at the cost of computing high-order approximations of a large number of polarizations. 
Fortunately, the cost of this computation can be somewhat reduced by leveraging the the DEM event classes of the EEGs in $G_{\textrm{EOC}}$ to selectively compute parts of the expansion of $\mathcal{E}_c$ that impact each polarization.

\subsection*{Detector Error Model Construction}
We now provide a full specification of our method for creating DEMs from error generator models. Given a circuit $C$, a set of detectors $\mathbb{D}$, and an error generator model for the layers of circuit $C$, the following procedure creates a DEM modeling the detection histories of $\Lambda(C)$.
\begin{enumerate}
    \item By deferring measurements, construct the expanded circuit \begin{align}
        \mathcal{C}' = \mathcal{E}_k\mathcal{U}(L'_k)\dots\mathcal{E}_1\mathcal{U}(L'_1),
    \end{align}
    where each $\mathcal{U}(L'_i)$ and $\mathcal{E}_i$ is an $(n+m)$-qubit superoperator that acts nontrivially on at most $n$ qubits.
    \item Propagate all errors to the end of the expanded circuit \begin{align}
        \mathcal{C}' = \mathcal{E}_k'\dots\mathcal{E}_1'\mathcal{C}_k\dots\mathcal{C}_1,
    \end{align} 
    then combine $\mathcal{E}_k'\dots\mathcal{E}_1'$ using a BCH expansion to get the approximate circuit error generator $\tilde{\mathcal{E}}_{C'} = \exp(G_{\textrm{EOC}})$.
    \item Decompose $G_{\textrm{EOC}}$ by DEM event class
    \begin{align}\label{eq:g_eoc_by_class}
        G_{\textrm{EOC}} = \sum_{G \in \mathbb{G}_{D_1}} \epsilon_G G + \sum_{G \in \mathbb{G}_{D_2}} \epsilon_G G + \cdots + \sum_{G \in \mathbb{G}_{D_k}} \epsilon_G G,
    \end{align}
    where $\mathbb{G}_{D_i}$ denotes the set of all EEGs in $G_{\textrm{EOC}}$'s composition in DEM event class $D_i$.
    \item Approximate the DEM using a Zassenhaus expansion:
    \begin{equation}
        \exp(G_{\textrm{EOC}}) \approx \prod\limits_{i=1}^{k} W_i,
    \end{equation}
    where $W_i$ denotes the $i$th Zassenhaus expansion term of $\exp{G_{\textrm{EOC}}}$ using the generator decomposition in Eq.~\eqref{eq:g_eoc_by_class}. Iteratively apply (3)-(4) to $W_i$ for $i>1$ to create single-DEM-event channels.
    \item Estimate the rate of each DEM event $D_i$. 
by estimating $p_{D_i}= \frac{1}{2}(1-\Tr(P\exp(G_{D_i})[\ketbra{\psi}]))$ for any representative $P \in D_i$. In limited cases (e.g. if $G_{D_i}$ is a purely S-type error generator), $p_{D_i}$ can be efficiently computed exactly. In other cases, use a low-order Taylor expansion of $\exp(G_{D_i})$ or similar approximation.
\end{enumerate}

Our method uses three approximations that can easily be tuned to higher accuracy (at the expense of higher computation time):
\begin{enumerate}
    \item A BCH approximation to compute $\tilde{\mathcal{E}}_c$.
    \item A Zassenhaus approximation to split $\tilde{\mathcal{E}}_c$ into single-event channels.
    \item An estimate of $\Tr(P\mathcal{E}[\ketbra{\psi}])$ for each single-event channel $\mathcal{E}$ to estimate the corresponding event rate. 
\end{enumerate}

In an error model with only S-type EEGs, this procedure constructs a DEM that is exact when approximations (1)-(2) are done to first order (because all S-type error generators commute), and the computation of (3) is done exactly (this can be done efficiently for a sparse model, see Appendix~\ref{app:theory}).

\subsection*{Sensitivity Analysis for Coherent Errors}
To quantify the relative impacts of different types of coherent error on cultivation, we computed polynomials relating the parameters of our error generator models to observables in cultivation circuits. Let $\boldsymbol{\theta}$ be a vector of all of parameters in the physical-level error generator model---for the spatially-homogeneous error models we study here, $\boldsymbol{\theta}$ consists of a rate for each single-qubit Hamiltonian EEG for each single-qubit gate, and a rate for each two-qubit Hamiltonian EEG for each two-qubit gate. We track how these physical-gate error parameters map onto observable properties of the circuit. The circuit error can be written as 
\begin{equation}
    \mathcal{E}_{c} = \exp(\sum_{P \in \mathbb{G}_H} \phi_P H_P),
\end{equation}
where $\mathbb{G}_H$ is a small set of Paulis (For a first-order BCH expansion, this set has size $O(gn_g)$, where $g$ is the number of gates in the circuit, and $n_g$ is the number of Hamiltonian EEGs needed to describe a single gate's error). For a first-order BCH approximation of $\mathcal{E}_{c}$, which is sufficient to reproduce the leading-order effects of coherent errors, the EEG error rate $\phi_Q$ is given by the vector $\phi_Q = \vb{s}_Q^T \boldsymbol\theta$ for some vector $\vb{s}_P$. Detector Pauli observables are then approximately
\begin{align}
    \langle P \rangle & \approx 1 + \boldsymbol\theta^T\left(\sum\limits_{Q,Q' \in \mathbb{G}_H} b_{P,Q,Q'}\vb{s}_{Q'}\vb{s}_Q^T \right)\boldsymbol\theta \\
    & = 1 + \boldsymbol\theta^T M_P \boldsymbol\theta
\end{align}
where $b_{P,Q,Q'}$ is a constant that encodes the impact of $H_QH_{Q'}$ on $\langle P \rangle$ for each pair $Q,Q' \in \mathbb{G}_H$. $M_P$ is a matrix that completely describes the sensitivity of $\langle P \rangle$ to physical-level coherent error parameters.

To approximate the relationship between the error model parameters and the discard probability for cultivation, we observe that at leading order, the discard probability is the the probability that a DEM events occurs---i.e., $\sum_{E \in E_1} p_{E}$, where $E_1$ is the set of DEM events in the leading-order DEM, and $p_E$ is the probability of event $E$. The probability of $E$ in the leading-order DEM can be expressed in terms of EEGs in the DEM event class corresponding to E as $p_E = -\frac{1}{2}\boldsymbol\theta^T M_{E} \boldsymbol\theta$, where 
%$M_{P,E} = \vb{s}_D B_P \vb{s}_D^{T}$, 
\begin{align}
    M_{E} = \sum\limits_{Q,Q' \in \mathbb{G}_E} b_{P,Q,Q'}\vb{s}_{Q'}\vb{s}_{Q}^T,
\end{align}
where $\mathbb{G}_E$ is the set of EEGs in $\mathcal{E}_c$ from DEM event class $E$ and $P \in E$ is a representative detector. This result allows us to construct a sensitivity matrix for the probability of a DEM event occurring in terms of the physical-level Hamiltonian error parameters, by summing over the individual event sensitivity matrices.
%, and $M_P$ is the matrix encoding the impact of terms $H_PH_{P'}$ (for $\Delta(P)=\Delta(P')=E$) on $\langle Q \rangle$, for a representative detector $Q \in E$.

\section*{Code Availability}

Code for our DEM creation algorithms is available in pyGSTi \cite{Nielsen2020-rd}. Code for our simulations will be made available in PRAQTICE \cite{praqtice}.

 %TC:ignore
\section*{Acknowledgments}

The authors thank Evangelia Takou, Kenneth Brown, Daniel Hothem, Alicia Magann, and Alireza Seif for helpful discussions. This material was funded in part by the U.S. Department of Energy, Office of Science, Office of Advanced Scientific Computing Research and by the Laboratory Directed Research and Development program at Sandia National Laboratories. This research is also based upon work supported by the Office of the Director of National Intelligence (ODNI), Intelligence Advanced Research Projects Activity (IARPA), specifically the ELQ program.  Sandia National Laboratories is a multi-program laboratory managed and operated by National Technology and Engineering Solutions of Sandia, LLC., a wholly owned subsidiary of Honeywell International, Inc., for the U.S. Department of Energy's National Nuclear Security Administration under contract DE-NA-0003525. All statements of fact, opinion or conclusions contained herein are those of the authors and should not be construed as representing the official views or policies of the U.S. Department of Energy or the U.S. Government.

\bibliography{bibliography}
\appendix
\onecolumngrid

\section{Theory for Creating DEMs for Arbitrary Small CPTP Errors}
\label{app:theory}
\subsection{Modeling and Propagating CPTP Errors in QEC Circuits}

We describe circuits as list of instructions to apply a series of unitary operations, $C = L_kL_{k-1}\dots L_1$, where operations are performed from right to left. Each constituent instruction $L_i$ is called a \emph{layer} of the circuit, typically corresponding to a set of logic operations performed simultaneously on a quantum processor. We will consider only circuits consisting of Clifford gates, stabilizer state preparation, and computational basis measurements, including mid-circuit measurements. For simplicity, in this work we assume that all mid-circuit measurements include post-measurement reset of the state in the state $\ket{0}$. This set of operations is sufficient for syndrome extraction circuits and circuits for logical Clifford gates in stabilizer codes. Our methods also apply to circuits with non-Clifford operations, at the expense of increased computational cost.

We model the error in a circuit layer with no measurements by a post-gate error channel. For layers with mid-circuit measurements, we model the error with post-layer and pre-layer error channels. We use this model only for simplicity of description; more general mid-circuit measurement noise can be introduced with the use of a gadget consisting of additional virtual qubits and a CNOT gate \cite{wysocki2026detailed}. We express each error channel as $\mathcal{E} = \exp(G)$, where $G$ is the error \emph{generator}. We express error generators in the \emph{elementary error generator} (EEG) basis \cite{blumekohout2022taxonomy}:
\begin{align}
    H_{P}\left[\rho\right]&=-i\left[P,\rho\right] ,\\
    S_{P}\left[\rho\right]&=P\rho P- \rho, \\
    C_{P,Q}\left[\rho\right]&=P\rho Q+Q\rho P-\frac{1}{2}\left\{\left\{P,Q\right\},\rho\right\},\\
    A_{P,Q}\left[\rho\right]&=i\left(P\rho Q-Q\rho P+\frac{1}{2}\left\{\left[P,Q\right],\rho\right\}\right).
\end{align}
where $P$ and $Q$ are $n$-qubit Paulis. The EEG basis has two critical features to the efficacy of our method. First, EEGs can be efficiently propagated through Clifford circuits \cite{miller2025simulation}. This results from the fact that for any unitary $U$,
\begin{align}
    \mathcal{U}^{\dagger}  S_{P}  \mathcal{U}\left[\rho\right]&= s_{U,P}S_{U^{\dag} PU}[\rho] \\
    \mathcal{U}^{\dagger}  H_{P}  \mathcal{U}\left[\rho\right]&= s_{U,P}H_{U^{\dag} PU}[\rho] \\
  \mathcal{U}^{\dagger}  C_{P,Q}  \mathcal{U}\left[\rho\right]&= s_{U,P}s_{U,Q}C_{U^{\dag} PU,U^{\dag} PU}[\rho] \\
  \mathcal{U}^{\dagger}  A_{P,Q}  \mathcal{U}\left[\rho\right]&= s_{U,P}s_{U,Q}A_{U^{\dag} PU,U^{\dag} PU}[\rho],
   \end{align}
where $s_{U,P}$ and $s_{U,Q}$ are signs and $\mathcal{U}$ is the superoperator corresponding to $U$. Propagating EEGs through Clifford circuits therefore reduces to conjugating the Pauli indices by Cliffords, which can be done efficiently. Second, errors in quantum processors are typically well-described by generators that are \emph{sparse} in the EEG basis---i.e., the number of basis elements appearing the in the error generator decomposition is usually much less than $16^n$ for an system of $n$ physical qubits. While the methods we introduce in this paper apply to \emph{any} CPTP error generator model, their cost scales polynomially with the sparsity of the model in the EEG basis, and therefore the whole circuit must be described by a relatively small number of EEGs for our method to be efficient.

We expect that in a system with H error rates of $O(\epsilon)$ are the rates of S, C, and A error generators are $O(\epsilon^2)$, so the contributions from $C$ terms that arise at second order from $H$ error generators should be considered along with first order $S$, $C$, and $A$ terms. Throughout this work, we will assume this hierarchy of error rates---which is equivalent to considering the leading-order contribution of each type of EEG to detection statistics.

\subsection{Error Propagation with Mid-Circuit Measurements}
Mid-circuit measurements are critical to QEC, but error generators cannot be propagated through mid-circuit measurements using the propagation rules for unitaries (because they are nonunitary operations). To propagate errors through a circuit with MCMs, we construct an \emph{expanded circuit}, a circuit that \emph{exactly} produces the detection histories of the original circuit but without mid-circuit measurements. We do this by deferring measurements to the end of the circuit, introducing auxiliary ``virtual'' qubits to the circuit (see Figure~\ref{fig:fig1}b) and modifying error channels to only act on the ``active'' qubits in the circuit at a given time. Every detector in the original circuit becomes a Pauli observable measured at the end of the expanded circuit. This basic approach is sufficient for MCMs with reset whose error is modeled by independent pre- and post-MCM error channels. Our method is also compatible with more sophisticated measurement error and non-resetting measurements, through the use of gadgets consisting of Clifford gates and virtual qubits \cite{wysocki2026detailed}.

Absorbing pre-measurement error channels into the previous circuit layer's error for notational simplicity, the imperfect implementation of $C$ is
\begin{equation}
\Lambda(C)=\mathcal{E}_{L_d}\mathcal{U}(L_d)\cdots \mathcal{E}_{L_1}\mathcal{U}(L_1).
\end{equation}
Propagating all errors to the end of the circuit, 
\begin{equation}
    \Lambda(C) = \mathcal{E}'_k\mathcal{E}'_{k-1}\dots \mathcal{E}'_1\mathcal{U}(C) = \mathcal{E}_c\mathcal{U}(C),
\end{equation}
where $\mathcal{U}(C)$ is the unitary superoperator for the perfect implementation of $C$, and we define
\begin{equation}
    \mathcal{E}_c = \mathcal{E}'_k\mathcal{E}'_{k-1}\dots \mathcal{E}'_1 = \exp(G_{c})
\end{equation}
where
\begin{equation}
    \mathcal{E}'_j = \mathcal{U}_{k}\dots\mathcal{U}_{j+1}\mathcal{E}_j \mathcal{U}_{j+1}^{-1} \dots \mathcal{U}_{k}^{-1}.
\end{equation}
We call $\mathcal{E}_c$ the \emph{circuit error channel} and $G_{c}$ the \emph{circuit error generator}, and they describe the dynamics of \emph{all} errors that occur in the circuit according to our error model. $G_{c}$ (as well as $\mathcal{E}_c$) is computationally hard to compute exactly, but we can efficiently approximate it using the Baker–Campbell–Hausdorff (BCH) formula \cite{BLANES2004135}, 
\begin{align}
    G_{c, j} &\approx \sum_{i=1}^k G'_i + \sum_{i=2}^{j} B_i,
\end{align}
where $B_j$ is the $j$th order BCH term. When $\mathcal{E}_1, \mathcal{E}_2, \dots, \mathcal{E}_d$ are in the EEG basis, the $G_{EOC, j}$ is also sparse---the number of EEGs needed for $G_{EOC, j}$ is a $j$th order polynomial in the number of EEGs needed to construct the physical-level, post-gate model of the circuit's error.

\subsection{Detector Error Models}

Our method constructs a \emph{detector error model (DEM)}, a coarse-grained model that tracks only the impact of errors on measurement results, not the errors' origin or form.  
A DEM models the outcomes of \emph{detectors}, Pauli observables $P$ that have definite outcome in the absence of errors ($\bra{\psi}P\ket{\psi}= \pm 1$), were $\ket{\psi}$ is the state ideally produced by the circuit. For notational simplicity, we use signed Pauli observables, and demand that $\bra{\psi}P\ket{\psi}= 1$ for all detectors $P$. 

For a circuit $C$, we specify a set of detectors $\mathbb{D}$, forming a basis for the set of commuting measurements that we want to model. There is freedom in choosing the detectors for a model, as long as they satisfy the criteria above. We typically define an ordering of the detectors $[P_1, P_2, \dots, P_{n_d}]$, and specify a DEM event as an $n_d$-bit string $b_1b_2\cdots b_{n_d}$, where $b_i=1$ if $P_i \in D$, and $b_i=0$ otherwise.
A \emph{detection history} is a record of the results of each detector measurement for a single run of a circuit---i.e., it is an $n_d$-bit string $b_1b_2\cdots b_{n_d}$, where $b_i=1$ if the result of measuring $P_i$, is $1$, and $b_i=0$ otherwise. 
A \emph{detector error model (DEM)} models the probability distribution of detection histories. It consists of a set of \emph{DEM events} $\{ D_1, D_2, \dots D_k\}$, with associated probabilities $\{p_1, p_2, \dots p_k\}$. A \emph{DEM event} is a subset $D \subseteq \mathbb{D}$ with an associated probability $p_D$.

DEMs are particularly convenient for modeling QEC circuits experiencing Pauli stochastic error, because Pauli faults can be directly mapped to DEM events \cite{gidney2021stim}. Any Pauli fault in a Clifford circuit can be efficiently propagated through the circuit to determine which measurement results it flips, and therefore what DEM event it causes. For a Pauli $Q$ and a set of detectors $\mathbb{D}$, we define a function $\toevent$ that maps the Pauli error $Q$ to the bit string $d$ encoding the DEM event caused by a Pauli $Q$ error immediately before measurement in the expanded circuit. Formally, $\toevent(Q) = D$ for $D \subset \mathbb{D}$, where for all $P \in \mathbb{D}$, $P \in D$ if and only if $[P,Q] \neq 0$.

\subsection{DEM Event Classes of Error Generators}
In this section, we review our method for constructing a DEM and show that our decomposition of error generators into single-event components [Eq.~\eqref{eq:decompose_generator}] leads to single-DEM-event error channels. 

The (approximate) end-of-circuit error channel for our circuit is given by 
\begin{equation}
    \mathcal{E}_c = \mathcal{E}'_k\mathcal{E}'_{k-1}\dots \mathcal{E}_1'
\end{equation}
where
\begin{equation}
    \mathcal{E}'_j = \mathcal{U}_{k}\dots\mathcal{U}_{j+1}\mathcal{E}_j \mathcal{U}_{j+1}^{-1} \dots \mathcal{U}_{k}^{-1}.
\end{equation}
This error channel is generally hard to compute exactly. Our approach is to approximate is using the $j$th order BCH expansion,
\begin{equation}
    \mathcal{E}_c = \exp(G_k')\exp(G_{k-1}')\cdots \exp(G_1') \approx \exp(\sum_{i=1}^k G_i)\exp(\frac{1}{2}\sum_{i=1}^k\sum_{j=i}^k [G_i, G_j])\cdots
\end{equation}
This error channel can be expressed in terms of elementary error generators as
\begin{equation}
    \mathcal{E}_c = \exp(\sum_{S_P \in \mathbb{G}_S} s_{p}S_P + \sum_{A_{P,Q} \in \mathbb{G}_A} a_{P,Q}A_{P,Q} + \sum_{C_{P,Q} \in \mathbb{G}_{C}} c_{P,Q}C_{P,Q} + \sum_{H_P \in \mathbb{G}_H} h_P H_P).
\end{equation}
Our goal is to relate this end-of-circuit error to DEM events it causes. Doing so is an approximation, because this channel models non-stochastic errors as well as stochastic errors, and a DEM is a purely stochastic model. However, in this section we show that the impacts of $\mathcal{E}_c$ on the circuit can be mapped onto DEM events to good approximation. Our approach is to decompose $G$ as a sum of error generators $G_1, G_2, \dots G_k$, where each $G_i$ contains all EEG terms belonging to a single DEM event class.

We then break $\mathcal{E}_c$ using the Zassenhaus formula. To leading order,
\begin{align}
    \exp(\sum_{G \in \mathbb{G}} \epsilon_G G) = \exp(\sum_{G \in \mathbb{G}_1} \epsilon_G G)\exp(\sum_{G \in \mathbb{G}_2} \epsilon_G G)\cdots \exp(\sum_{G \in \mathbb{G}_k} \epsilon_G G) +O(\epsilon^{\nicefrac{3}{2}}) \label{eq:event_zassenhaus}
\end{align}
Where $\mathbb{G}_1, \mathbb{G}_2, \dots, \mathbb{G}_k \subseteq \mathbb{G}$ are the DEM event classes of EEGs.We can also consider higher-order Zassenhaus expansions to account for the error in the DEM event channel approximation. These corrections involve commutators of the single-DEM-event generators, and can lead to additional events in the DEM.
The second order Zassenhaus expansion term is 
\begin{align}
    \exp\left(\frac{1}{2}\left(\sum_{g_1 \in G_1, g_2 \in G_2} \epsilon_{g_1}\epsilon_{g_2}[g_2, g_1] + \sum_{g_1 \in G_1, g_3 \in G_3} \epsilon_{g_1}\epsilon_{g_3}[g_3, g_1] + \cdots + \sum_{g_{k-1} \in G_{k-1}, g_k \in G_k} \epsilon_{g_{k-1}}\epsilon_{g_k}[g_{k-1}, g_k]\right)\right).
\end{align}
This error channel will typically contain terms from multiple DEM event classes. To account for this correction term in our DEM, we decompose it into DEM event channels using the first-order Zassenhaus expansion again. In principle, this process of breaking up channels into DEM event channels can be continued to arbitrary order.

\begin{definition}
The \emph{DEM event class} of DEM event $D$ is the set of all EEGs $G$ such that $D$ is the largest subset of $\mathbb{D}$ such that for all $Q \in D$ and all Pauli indices $P$ of $G$, $[P,Q] \neq 0$. We use $\eegev(G)$ to denote the DEM event $D$ corresponding to $G$'s DEM event class.
\end{definition}

We now state and prove Theorem~\ref{thm:eeg_to_event}, which motivates decomposing the circuit error generator by DEM event class.  
\vspace{8pt}

\noindent
\textbf{Theorem 1.} \textit{Let $G$ be an error generator, and suppose its EEG decomposition consists of entirely EEGs of the same DEM event class $D$. Then the circuit error $\mathcal{E}_D = \exp(G)$ is perfectly modeled by the DEM with event $D$ with probability $\frac{1}{2}(1-\Tr(P\mathcal{E}_D[\ketbra{\psi}]))$, where $P \in D$.}

\vspace{8pt}
The first key idea of Theorem~\ref{thm:eeg_to_event} is that each EEG can be mapped in a well-defined way to a \emph{single} DEM event. The relation is formalized in the following lemma. 

\begin{lemma}\label{lem:eeg_to_event}
Let $G$ be an elementary error generator. The effect of a small error $1 + \epsilon_G G$ on the detection statistics for a set of detectors $\mathbb{D}$ is modeled exactly by one DEM event, as follows:
\begin{itemize}
    \item If $G=S_P$, the DEM for $1 + \epsilon G$ is $\{\toevent(P): \epsilon_G\}$
    \item If $G=C_{P,P'}$ the DEM is only nontrivial if $\toevent(P)=\toevent(P')$ and $[P,P'] = 0$. If it is nontrivial, it is given by $\{\toevent(P): \pm 2\epsilon_G\}$ 
    \item If $G=A_{P,P'}$ the DEM is only nontrivial if $\toevent(P)=\toevent(P')$ and $[P,P'] \neq 0$. If it is nontrivial, it is given by $\{\toevent(P): \pm 2\epsilon_G\}$ 
    \item If $G = H_{P}$, the DEM has no events. A small error $1+\frac{1}{2}\epsilon_{H_P}\epsilon_{H_P'}H_PH_{P'}$ has a single-event DEM given by the $C_{P,P'}$ result above.
\end{itemize}
\end{lemma}

\vspace{8pt}
\noindent
\textbf{Proof of Lemma~\ref{lem:eeg_to_event}:} Our proof proceeds by considering the action of EEGs on detector Paulis $\langle{P}\rangle$, which is determined by computing $\Tr(PG[\psi])$ for each $P \in \mathbb{D}$. We show that for each $P \in \mathbb{D}$, either $\Tr(PG[\ketbra{\psi}]) = 0$ or $\Tr(PG[\ketbra{\psi}]) = k_{G,\psi}\bra{\psi} P \ket{\psi}$, where $k_{G,\psi}$ is a constant that is independent of $P$. For a detector $P$, the error $1+\epsilon G$ gives the expectation value
\begin{equation}
    \langle P \rangle =  \langle P \rangle_{0}  +  \beta(\psi, G, P) \epsilon_{G}
\end{equation}
where $ \langle P \rangle_{0} = \textrm{Tr}(P\ket{\psi}\bra{\psi})$ and $\beta(\psi, G, P) = \textrm{Tr}(PG[\ket{\psi}\bra{\psi}])$. The values of $\beta(\psi, G, P)$ determine how each error generator impacts the detector Pauli expectation values $\langle P \rangle$, and we explicitly compute them here.

First, we consider Pauli stochastic error generators $S_Q$. We have that 
\begin{equation}
   \beta(\psi, S_Q, P) =  \begin{cases} \mp 2\langle P \rangle_0 & \textrm{if } [P,Q] \neq 0\\
    0  & \textrm{else}\\
    \end{cases}
\end{equation}
If $\Delta(Q)=\emptyset$, then $[P,Q]=0$ for any detector $P$, so $\beta(\psi, S_Q, P)=0$. Otherwise, for $P \in \Delta(Q)$, $\beta(\psi, S_Q, P)=-2$  (because we have assumed $\langle P \rangle_0=1)$.

Now, we consider the case of an Pauli correlation error generator $C_{Q_1, Q_2}$ with $\Delta(Q_1)=\Delta(Q_2)$, and compute its effect on a detector $P$, where $P \in \Delta(Q_1)$. 

\begin{align}
    \Tr(PC_{Q_1, Q_2}[\ketbra{\psi}]) & = \bra{\psi} Q_1PQ_2 \ket{\psi} + \bra{\psi} Q_2PQ_1 \ket{\psi} - \frac{1}{2} \bra{\psi}\{\{Q_1, Q_2\},P\} \ket{\psi})
\end{align}
Because $P$ is a detector, 
\begin{align}
    \Tr(PC_{Q_1, Q_2}[\ketbra{\psi}]) & = \bra{\psi} Q_1PQ_2 \ket{\psi} + \bra{\psi} Q_2PQ_1 \ket{\psi} - \langle P \rangle_0  \bra{\psi}\{Q_1, Q_2\} \ket{\psi})
\end{align}
If either $[P,Q_1]=0$ or $[P,Q_2]=0$,
\begin{align}
    \Tr(PC_{Q_1, Q_2}[\ketbra{\psi}]) & = \langle P \rangle_0 \bra{\psi} Q_1Q_2 \ket{\psi} + \langle P \rangle_0 \bra{\psi} Q_2Q_1 \ket{\psi} - \langle P \rangle_0  \bra{\psi}\{Q_1, Q_2\} \ket{\psi})=0
\end{align}
Otherwise, 
\begin{align}
    \Tr(PC_{Q_1, Q_2}[\ketbra{\psi}]) & = -\langle P \rangle_0 \bra{\psi} Q_1Q_2 \ket{\psi} - \langle P \rangle_0 \bra{\psi} Q_2Q_1 \ket{\psi} - \langle P \rangle_0 \frac{1}{2} \bra{\psi}\{Q_1, Q_2\} \ket{\psi}) = -2\langle P \rangle_0  \bra{\psi}\{Q_1, Q_2\} \ket{\psi}
\end{align}
From which we conclude that
\begin{equation}
   \beta(\psi, C_{Q_1, Q_2}, P) =  \begin{cases} \pm 4\langle P \rangle_0 & \textrm{if } [Q_1,Q_2]=0, [Q_1,P]\neq 0, [Q_2,P]\neq 0,\textrm{ and } Q_1Q_2\ket{\psi}=\mp \ket{\psi}\\
    0  & \textrm{else}\\
    \end{cases} 
\end{equation}

If $\toevent(Q_1)=\emptyset$ and $\toevent(Q_2)\neq \emptyset$ then for any detector Pauli $P$, $[P,Q_1]=0$. If $P \in \toevent(Q_2)$, then $[P,Q_2] \neq 0$, 
It suffices to consider the case where $\toevent(Q_1)\neq \emptyset$ and $\toevent(Q_2)\neq \emptyset$ (because otherwise, $[P,Q_1]=0$ or $[P,Q_2]=0$, which implies that $\beta(\psi, C_{Q_1, Q_2}, P) = 0$). In this case,
\begin{equation}
   \beta(\psi, C_{Q_1, Q_2}, P) =  \begin{cases} \pm 4\langle P \rangle_0 & \textrm{if } [Q_1,Q_2]=0 \textrm{ and } Q_1Q_2\ket{\psi}=\mp \ket{\psi}\\
    0  & \textrm{else}\\
    \end{cases} 
\end{equation}
We see that $\delta(\psi, Q_1, Q_2, P)$ is independent of the detector $P$ within the set $\toevent(Q_1)=\toevent(Q_2)$. This means that $1+\epsilon_{C_{Q_1,Q_2}}C_{Q_1,Q_2}$ affects all such detectors $P$ the same amount and in the same direction.

We now consider the case of a active error generator $A_{Q_1, Q_2}$ with $\Delta(Q_1)=\Delta(Q_2)$, and compute its effect on a detector $P \in \Delta(Q_1)$.
\begin{align}
    \Tr(PA_{Q_1, Q_2}[\ketbra{\psi}]) & = i\bra{\psi} Q_2PQ_1 \ket{\psi} - i\bra{\psi} Q_1PQ_2 \ket{\psi} + \frac{i}{2} \bra{\psi}\{P,[Q_1, Q_2]\} \ket{\psi}
\end{align}
Because $P$ is a detector, 
\begin{align}
    \Tr(PA_{Q_1, Q_2}[\ketbra{\psi}]) & = i\bra{\psi} Q_2PQ_1 \ket{\psi} - i\bra{\psi} Q_1PQ_2 \ket{\psi} + i \langle P \rangle_0  \bra{\psi}[Q_1, Q_2] \ket{\psi}
\end{align}
If either $[Q_1, P]=0$ or $[Q_2, P]=0$, then
\begin{align}
    \Tr(PA_{Q_1, Q_2}[\ketbra{\psi}]) & = i\langle P \rangle_0\bra{\psi} Q_2Q_1 \ket{\psi} - i\langle P \rangle_0\bra{\psi} Q_1Q_2 \ket{\psi} + i \langle P \rangle_0  \bra{\psi}[Q_1, Q_2] \ket{\psi}=0
\end{align}
Otherwise, 
\begin{align}
    \Tr(PA_{Q_1, Q_2}[\ketbra{\psi}]) & = -i\langle P \rangle_0\bra{\psi} Q_2Q_1 \ket{\psi} + i\langle P \rangle_0\bra{\psi} Q_1Q_2 \ket{\psi} + i \langle P \rangle_0  \bra{\psi}[Q_1, Q_2] \ket{\psi}=2 i \langle P \rangle_0  \bra{\psi}[Q_1, Q_2] \ket{\psi},
\end{align}
from which we conclude that 
\begin{equation}
   \beta(\psi, A_{Q_1, Q_2}, P) =  \begin{cases} \pm 4\langle P \rangle_0 & \textrm{if } [Q_1, P]\neq 0, [Q_2, P]\neq 0, [Q_1,Q_2]\neq 0 \textrm{ and } iQ_1Q_2\ket{\psi}=\pm \ket{\psi}\\
    0  & \textrm{else}\\
    \end{cases} 
\end{equation}
In this case, if $\Delta(Q_1)=\emptyset$ or $\Delta(Q_2)=\emptyset$, then $\delta(\psi, Q_1, Q_2, P)=0$, because either $[Q_1, P] = 0$ or $[Q_2, P] = 0$. In the case where $\Delta(Q_1)=\emptyset$ or $\Delta(Q_2)\neq\emptyset$, the magnitude of $\delta(\psi, Q_1, Q_2, P)$ and its sign relative to $\langle P \rangle_0$ is independent of $P$ (analogously to the case of $C$-type EEGs above), so for all $P \in \Delta(Q_1)$, $\langle P \rangle$ is perturbed in the same way by $A_{Q_1, Q_2}$.

Finally, we consider $H$-type EEGs. These cannot contribute at first order to a change in $\langle P \rangle$ for $P \in \mathbb{D}$. This is because
\begin{equation}
   \beta(\psi, H_{Q}, P) =  \begin{cases} \pm 2\langle P \rangle_0 & \textrm{if } [Q, P]\neq 0, iPQ = \mp \ket{\psi}\\
    0  & \textrm{else}\\
    \end{cases}, 
\end{equation}
however, $P\ket{\psi} = \ket{\psi}$, so $QP\ket{\psi} = \ket{\psi}$ implies $Q\ket{\psi} = \pm \ket{\psi}$, which in turn requires $[Q,P]=0$.
To determine the leading-order contribution of $H$-type EEGs to a DEM, we consider an error $1+\frac{1}{2}\epsilon_{H_P}\epsilon_{H_P'}H_PH_{P'}$. We have 
\begin{equation}
   H_{Q_1}H_{Q_2} =  \begin{cases} \-iH_{Q_1Q_2} + C_{Q_1, Q_2} & \textrm{if }  [Q_1,Q_2]= 0,\\
    C_{Q_1, Q_2} & \textrm{if }  [Q_1,Q_2] \neq 0\\
    \end{cases} 
\end{equation}
$H_{Q_1Q_2}$ has no impact on definite-outcome Pauli measurements, so we need only analyze the $C_{Q_1, Q_2}$ terms, which we have already done above. \qed

These computations show that for a fixed EEG $G$, any detector Pauli $P$ changes satisfies $\beta(\psi, G, P)=0$ or $\beta(\psi, G, P)=k$ for some $k$. We now show that the resulting structure of Pauli expectation values implies that the error $1+\epsilon G$ corresponds to a single-event DEM. To do this, we first use the estimated values of $\langle P \rangle$ for each detector $P$ to reconstruct the expectation values for all products of detector Paulis. 

\begin{lemma}\label{lem:beta_detector_products}
    Let $D$ be a set of detectors. For any $P, P'$ such that $P$ and $P'$ are either in $D$ or are products of elements of $D$, the following holds:
    \begin{enumerate}
        \item If $\beta(\psi, G, P) = k $ and $\beta(\psi, G, P') = k $, then $\beta(\psi, G, PP') = 0$.
        \item If $\beta(\psi, G, P) = k $ and $\beta(\psi, G, P') = 0$, then $\beta(\psi, G, PP') = k$.
    \end{enumerate}
\end{lemma}

\vspace{8pt}
\noindent
\textbf{Proof of Lemma~\ref{lem:beta_detector_products}:} 
We start with the case of stochastic Pauli error generators. Consider $G=S_Q$ with $\beta(\psi, G, P) = k $ and $\beta(\psi, G, P') = k $. We have that $[P,Q] \neq 0$ and $[P',Q] \neq 0$, and therefore $[Q, PP'] = P[Q, P']+[Q,P]P' = 2PQP' + 2QPP' = 0$, so $\beta(\psi, G, PP') = 0$. Similarly, if $\beta(\psi, G, P) = k $ and $\beta(\psi, G, P') = 0$, then $[P,Q] \neq 0$ and $[P',Q] = 0$, which implies that $[Q, PP'] = P[Q, P']+[Q,P]P' = [Q,P]P' \neq 0$. Furthermore, $Q\ket{\psi}= \ket{\psi}$ (because $\beta(\psi, G, P) = k $). This implies that $\beta(\psi, G, PP') = k$.

We now consider the cases of Pauli-correlation and active error generators $G$. Let $G$ be indexed by Paulis $Q_1, Q_2$. $\beta(\psi, G, P) = k $ implies that $[Q_1, P] \neq 0$ and $[Q_2, P] \neq 0$, and $\beta(\psi, G, P') = k $ implies that $[Q_1, P'] \neq 0$ and $[Q_2, P'] \neq 0$. Because $P$ and $P'$ are detectors, $[P,P']=0$, and therefore $[Q_1, PP'] = P[Q_1, P']+[Q_1,P]P' = 2PQ_1P' + 2Q_1PP' = 0$. Similarly, $[Q_2, PP']=0$ Therefore, $\beta(\psi, G, P) = 0$. Similarly, if $\beta(\psi, G, P) = k $ and $\beta(\psi, G, P') = 0$, then if $G$ is a $C$ or $A$ term, we know that $Q_1Q_2\ket{\psi} = \pm \ket{\psi}$ and (for $C$-type EEGs) $[Q_1, Q_2]=0$ or (for $A$-type EEGs) $[Q_1, Q_2] \neq 0$. Furthermore, it must be that either $[Q_1,P']=0$ or $[Q_2,P']=0$. This implies that $[Q_1, PP'] = P[Q_1, P']+[Q_1, P]P' = [Q_1, P]P' \neq 0$, and similarly $[Q_2, PP'] \neq 0$. Therefore, $\beta(\psi, G, PP') = k$. \qed

\vspace{8pt}

Applying Lemma~\ref{lem:beta_detector_products} allows us to reconstruct the full vector of Pauli expectation values for $\mathcal{P}(D)$. We have that any Pauli $P$ that is a product of an even number of detectors has $\beta(\psi, G, P) = 0$, and any Pauli $P$ that is a product of an odd number of detectors has $\beta(\psi, G, P) = k$. We use these values to compute the rate of each DEM event. This completely determines the polarizations $\langle \prod_{P \in D'} P\rangle$, and it's straightforward to see that a DEM with the single event $D$ with probability $-\nicefrac{k}{2}$ exactly reproduces these polarizations (which fully describe the detector outcome distribution). To see this, note that the polarization of $\langle \prod_{P \in D'} P\rangle$ for $D' \subseteq D$ with $|D'|$ odd is $1-2p_{D'} = 1+k$, and the polarization of $\langle \prod_{P \in D'} P\rangle$ for $D' \subseteq D$ with $|D'|$ even is $1$.

We conclude that the error $1+\epsilon_G G$ causes a change in exactly one DEM event probability, proving Lemma~\ref{lem:eeg_to_event}. \qed

\vspace{8pt}
\noindent
\textbf{Proof of Theorem~\ref{thm:eeg_to_event}:} We  use Lemma~\ref{lem:eeg_to_event} to prove Theorem~\ref{thm:eeg_to_event}. For an error generator whose EEG representation is $\sum_{G \in \mathbb{G}} \epsilon_G G$ where all terms $G \in \mathbb{G}$ correspond to the same DEM event, its exponentiation can be expressed via Taylor expansion in terms of compositions of error generators,
\begin{equation}
    \exp(\sum_{G \in \mathbb{G}} \epsilon_G G) = 1 + \sum_{G \in \mathbb{G}} \epsilon_G G + \sum_{G,G' \in \mathbb{G}} \epsilon_G\epsilon_{G'} GG' + \dots \label{eqn:taylor_expand_class}
\end{equation}
All Pauli indices $P$ in the Taylor expanded error channel is a product of EEGs whose indices $P$ all satisfy $\Delta(P) = D$ for fixed $D$ (following the composition rules in Ref~\cite{miller2025simulation}). To complete our proof, we require one final technical lemma. 
\begin{lemma}\label{lem:p_or_i} For any Pauli operators $Q$ and $Q'$, if $\toevent(Q)=\toevent(Q')$, then $\toevent([Q,Q'])$, $\toevent(\{ Q, Q'\})$, and $\toevent(QQ')$ are all either $\toevent(Q)$ or $\toevent(I)$. 
\end{lemma}

\vspace{8pt}
\noindent
\textbf{Proof of Lemma~\ref{lem:p_or_i}:} Let $P \in \mathbb{D}$ be a detector Pauli. For a Pauli $Q$, $P \in \Delta(Q)$ if $[Q,P] \neq 0$. Define $\eta_{P,P'}$ for $P, P' \in \pauli_n$ as $\eta_{P,P'}=0$ if $[P,P'] = 0$ and  $\eta_{P,P'}=2$ if $[P,P'] \neq 0$. We have that
\begin{equation*}
[QQ',P] = [Q,P]Q'+Q[Q',P]=\eta_{Q,P}QPQ'-\eta_{Q',P_s}QPQ'=0,
\end{equation*}
where the last line follows because $P \in \Delta(Q), P \in \Delta(Q')$ (which implies $\eta_{Q,P}=\eta_{Q',P}$). Therefore, $\Delta(QQ')=0$.  We also have that $[[Q,Q'],P] = [QQ',P]-[Q'Q,P]=0$ and $[\{Q,Q'\},P] = [QQ',P]+[Q'Q,P]=0$ using the result above.  \qed

\vspace{8pt}

It follows from Lemma~\ref{lem:p_or_i} that every EEG in the decomposition of the Taylor expanded error has indices satisfying $\Delta(P) = D$ or $\Delta(P) = \emptyset$. Therefore, $\exp(\sum_{G \in \mathbb{G}} \epsilon_G G)$ is a sum of error generator terms from DEM event classes $D$ or $\emptyset$, which implies, via application of Lemma~\ref{lem:eeg_to_event} and linearity, that it can be completely described by a DEM with only event $D$. \qed

\vspace{8pt} 

We have shown that if $\mathcal{E} = \exp(\sum_{G \in \mathbb{G}} \epsilon_G G)$, where $\mathbb{G}_1$ consists of only EEGs from a single DEM event class, then $\mathcal{E}$ is an error channel that causes one DEM event. The probability of this event is $p_D = \Tr(P\mathcal{E}[\ketbra{\psi}])$, where $P$ is a representative detector from the DEM event. The value of $p_D$ can be estimated by approximating this matrix exponential, e.g., using a Taylor expansion.

\subsection{S-Only error models}

Here, we show how to exactly compute the DEM event rate for a single-event channel with only $S$-type EEGs. Letting $Q_i$ be a representative Pauli for DEM event $i$, the rate of DEM event $i$ is 
\begin{align}
    \frac{1}{2}\left(1-\Tr(Q_i \exp(\sum_{S_j \in G_i}\epsilon_i S_j)[\ketbra{\psi}])\right) & = \frac{1}{2}\left( 1 - \exp(\sum_{S_j \in G_i} 2\epsilon_j)\right),
\end{align}
where $\zeta$ is a normalization constant. We then have
\begin{align}
    \frac{1}{2}\left(1-\Tr(Q_i \exp(\sum_{S_j \in G_i}\epsilon_i S_j)[\ketbra{\psi}])\right) & = \frac{1}{2}\left( 1 - \Tr(Q_i\mathcal{E}_{Q_i,Q_i}[Q_i])\right) \\
    & = \sum_{Q':[Q',Q]\neq0} p_{Q'},
\end{align}
where $p_{Q'}$ denotes the probability of a Pauli $Q'$ error in the error channel $\mathcal{E}_c$. 

\section{Detection Statistics in Error Generator Models}
\label{app:cptp}
In this appendix, we consider the ability of DEMs to approximate detector statistics generated by CPTP errors. We first prove that there is always a DEM with nonnegative probabilities that faithfully reproduces detection statistics at leading order. We then give an example where a DEM would need a negative-probability event to reproduce statistics beyond leading order. 

\subsection{Reproducing CPTP Dynamics with DEMs}
Error generator models are capable of modeling a much richer space of dynamics than DEMs, which model error events as stochastic and independent. However, as we show here, a DEM is capable of accurately modeling the probabilities of detection histories at first order in $S$, $C$ and $A$ error rates, and second order in $H$ error rates. 

\begin{theorem}\label{thm:cptp_dem}
    A CPTP model cannot generate a negative DEM event rate at leading order in EEG rates.
\end{theorem}

First we prove a few useful lemmas on the impact of $A$ and $C$ EEGs on detector Pauli observables. 

\begin{lemma}\label{lem:c_terms}
    Suppose that $\beta(Q, C_{P_1,P_2}, \psi) = 4\gamma_{1,2}$ and $\beta(Q, C_{P_3,P_2}, \psi) = 4\gamma_{1,3}$ for a detector $Q$ and $\gamma_{1,2}, \gamma_{1,3} \in \{1, -1\}$. Then  $\beta(Q, C_{P_1,P_3}, \psi) = 4\gamma_{1,2}\gamma_{1,3}$.
\end{lemma}

\vspace{8pt}
\noindent
\textbf{Proof of Lemma~\ref{lem:c_terms}:} For any $C_{P,P'}$, $\beta(Q, C_{P,P'}, \psi)= \pm 2$ implies that $[Q, P'] \neq 0$, $[Q,P] \neq 0$, and
$PP'\ket{\psi}=\mp\ket{\psi}$, 
which implies that $\beta(Q, S_P, \psi)=\beta(Q, S_{P'}, \psi)=-2$. We also have that
\begin{equation}
    \Tr(Q C_{P_1,P_j}[\ketbra{\psi}]) =  \begin{cases}  \pm 4 \textrm{  if  } [P_1, P_j] = 0 \textrm{  and } P_1P_j\ket{\psi} = \gamma_{1,j}\ket{\psi} \\ 0 \textrm{ else } \end{cases},
\end{equation}
Because $[P_1, P_2] = 0$, $[P_1, P_3] = 0$, $P_1P_2 = \gamma_{1,2}\ket{\psi}$, and $P_1P_3 = \gamma_{1,3}\ket{\psi}$, we have that $P_2P_3 = P_2P_1P_1P_3\ket{\psi} = P_1P_2P_1P_3\ket{\psi} =  \gamma_{1,2}\gamma_{1,3}\ket{\psi}$. To show that $C_{P_2,P_3}$ contributes, it remains to be shown that $[P_2,P_3]=0$. We have that $P_1P_3P_1P_2\ket{\psi} = \gamma_{1,3}\gamma_{2,3} \ket{\psi}$ and that, by our assumption $[P_3,P_1]=0$. So, we have that $P_3P_2\ket{\psi} = \gamma_{1,3}\gamma_{2,3} \ket{\psi} = P_2P_3\ket{\psi}$, which implies that $[P_2,P_3]=0$. \qed

\begin{lemma}\label{lem:a_terms}
    Suppose that $\beta(Q, A_{P_1, P_2}, \psi)\neq 0$ and $\beta(Q, A_{P_3, P_2}, \psi) \neq 0$ for a detector $Q$. Then $\beta(Q, A_{P_1,P_3}, \psi) = 0$.
\end{lemma}

\vspace{8pt}
\noindent
\textbf{Proof of Lemma~\ref{lem:a_terms}:} Suppose $A_{P_1, P_2}$ and $A_{P_3, P_2}$ contribute to $\langle Q \rangle$ nontrivially. Then $iP_1P_2 = \gamma_{1,2}\ket{\psi}$ and $iP_{3}P_{2} = \gamma_{3,2}\ket{\psi}$ for signs $\gamma_{1,2}$ and $\gamma_{3,2}$, and $\{P_1,P_2\}=\{P_1,P_3\}=0$. Then we have $P_1P_3\ket{\psi}= P_1P_2P_2P_3\ket{\psi} = -P_1P_2P_3P_2\ket{\psi}  =\gamma_{1,2}\gamma_{3,2}\ket{\psi}$. Additionally, $\gamma_{1,2}\gamma_{3,2}\ket{\psi} = -P_3P_2P_1P_2\ket{\psi}  = P_3P_2P_2P_1\ket{\psi} =P_3P_1\ket{\psi}$. If $[P_1,P_3] \neq 0$, then it follows that $P_1P_3\ket{\psi}=-\gamma_{1,2}\gamma_{3,2}\ket{\psi}$, which is a contradiction, so $[P_1,P_3]=0$, and therefore $A_{P_1,P_3}$ does not contribute. \qed

\vspace{8pt}

We now use Lemma~\ref{lem:c_terms} and Lemma~\ref{lem:a_terms} to prove the Theorem~\ref{thm:cptp_dem}. 

\vspace{8pt}
\noindent
\textbf{Proof of Theorem~\ref{thm:cptp_dem}.} To prove this result, we need only consider terms in the leading-order approximation to the circuit error that impact a particular detector Pauli $Q$ (we can apply the same reasoning to each detector). Without lose of generality, we will use error generators with signed Pauli indices, and assume that all error generator rates are positive. The leading order Taylor expansion of such an error channel is
\begin{align}
    \mathcal{E} & \approx 1 + \left(\sum_{S_P \in \mathbb{G}_S} s_{p}S_P + \sum_{C_{P,Q} \in \mathbb{G}_{C}} c_{P,Q}C_{P,Q} + \sum_{A_{P,Q} \in \mathbb{G}_{A}} a_{P,Q}A_{P,Q} + \sum_{H_P \in \mathbb{G}_H} h_P H_P + \frac{1}{2}\sum_{H_P,H_{P'} \in \mathbb{G}_H} h_Ph_{P'} (H_PH_{P'} + H_{P'}H_{P})\right) \\
    & \approx 1 + \left(\sum_{S_P \in \mathbb{G}_S} s_{p}S_P + \sum_{C_{P,Q} \in \mathbb{G}_{C}} c_{P,Q}C_{P,Q} + \sum_{A_{P,Q} \in \mathbb{G}_{A}} a_{P,Q}A_{P,Q} + \sum_{H_P \in \mathbb{G}_H} h_P H_P + \frac{1}{2}\sum_{H_P,H_{P'} \in \mathbb{G}_H}   h_Ph_{P'} C_{P,P'} - \frac{i}{2}\sum_{H_P,H_{P'} \in \mathbb{G}_H, [P,P]=0} h_Ph_{P'} H_{PP'} \right) , \label{eqn:leading_order_taylor}
\end{align}
where we have decomposed the products of $H$ terms into $C$ and $H$ terms.
The DEM event rate is $\langle Q \rangle =   \Tr(Q\mathcal{E}[\ketbra{\psi}])$, and to leading order this is
\begin{align}
    \Tr(Q\mathcal{E}[\psi]) & \approx \langle Q \rangle_0 + \left(\sum_{S_P \in \mathbb{G}_S} s_{p}\beta(Q, S_P, \psi) + \sum_{C_{P,Q} \in \mathbb{G}_{C}} c_{P,Q}\beta(Q, C_{P,Q}, \psi) +  \sum_{A_{P,Q} \in \mathbb{G}_{A}} a_{P,Q}\beta(Q, A_{P,Q}, \psi)  + \frac{1}{2}\sum_{H_P,H_{P'} \in \mathbb{G}_H}   h_Ph_{P'} \beta(Q, C_{P,P'}, \psi) \right) , \label{eqn:leading_order_taylor2}
\end{align}
where we have used the fact that $\beta(Q,H_P, \psi)=0$ for all $H_P$. As $\langle Q \rangle_0 = 1$ and we are in the small errors regime, we therefore only need to show that the remaining terms in Eq.~\eqref{eqn:leading_order_taylor2} sum to a value that is $\leq 0$.

First, we compute the contribution for the final terms in this sum, which arise from the quadratic $H$ terms, i.e., the $
\frac{1}{2}\sum_{H_P,H_{P'} \in \mathbb{G}_H} h_Ph_{P'} \beta(Q, C_{P,P'}, \psi)$ term.
Recall that $\beta(Q, C_{P,P'}, \psi) = \pm 2  \textrm{ if } [P,P']=0, [P,Q]\neq 0, [P',Q]\neq 0,\textrm{ and } PP'\ket{\psi}=\mp \ket{\psi}$. Suppose $(P_1, P_2)$ and $(P_2, P_3)$ satisfy the aforementioned commutation conditions, with $P_1P_2\ket{\psi} = P_2P_3\ket{\psi} = -\ket{\psi}$. Then either:
\begin{enumerate}
    \item $[P_1,P_3] \neq 0$, in which case $\beta(Q, C_{P,P'}, \psi) = 0$
    \item $[P_1,P_3] = 0$, in which case $P_1P_3\ket{\psi} = P_1P_2P_2P_3\ket{\psi} = \ket{\psi}$, which means that $\beta(Q, C_{P_1,P_3}, \psi) = - 2$.
\end{enumerate} 
There can be at most $\nicefrac{\abs{\mathbb{G}_H}}{2}$ pairs of Pauli indices $(P,P')$ where $H_P, H_{P'} \in G$ with the properties (1) $\beta(Q, C_{P,P'}, \psi) = 2$ and (2) no Pauli occurs in more than one pair.  However, there are $\abs{\mathbb{G}_H}$ pairs $(P,P)$ (where $H_P \in G$), and for all such pairs, $\beta(Q, C_{P,P}, \psi) = -2$. Therefore, 
\begin{equation}
    \frac{1}{2}\sum_{H_P,H_{P'} \in \mathbb{G}_H}   h_Ph_{P'} \beta(Q, C_{P,P'}, \psi) \leq -\frac{\abs{\mathbb{G}_H}}{2}.
\end{equation}
This result also implies that $\tilde{\langle Q\rangle } \leq 1$ at leading order for $H$-only models.

We now consider the contribution of the $S$, $C$ and $A$ terms, for which it convenient to set $h_P=0$.
The CPTP condition requires that $\mathbb{G}_C \subseteq \{C_{P,P'} \mid S_P, S_{P'} \in \mathbb{G}_S\}$ (and furthermore, note that $C_{P,P'} \in \mathbb{G}_C$ implies that $S_P, S_{P'} \in \mathbb{G}_S$, which is required for the circuit error to be CPTP). 
We have that $\Tr(Q S_P [\psi]) = -2$ for all $S_P \in \mathbb{G}_S$, as we have restricted to terms where $[Q,P] \neq 0$. We arrange the error rates for the set $\mathbb{S}_{Q}$ into a matrix $M$, where $M_{P,P'} = c_{P,P'}+ia_{P,P'}$ for the upper triangle,  $M_{P,P'} = c_{P',P}-ia_{P,P'}$ for the lower triangle, and  $M_{P,P} = s_{P}$ on the diagonal, with $P, P' \in \mathbb{S}_Q$. 
We have that 
\begin{equation}
\Tr(Q\mathcal{E}[\psi])  = 1 +\Tr(\sum_{P \in \mathbb{S}_Q} s_p S_P + \sum_{P,P' \in \mathbb{S}_Q} c_{P,P'}C_{P,P'} + \sum_{P,P' \in \mathbb{S}_Q} a_{P,P'}A_{P,P'}),
\end{equation}
and we now show that we can write this as 
\begin{equation}\label{eq:matrix_form}
\Tr(Q\mathcal{E}[\psi])  = 1  -2 \boldsymbol{x}^{\dagger} M\boldsymbol{x}
\end{equation}
where
\begin{equation}
x_P = -i\kappa_{P,P'}+\kappa'_{P,P'}
\end{equation}
and 
\begin{align}
    \kappa_{P,P'} &= \begin{cases} -\textrm{sgn}(\beta(Q, C_{P,P'}, \psi)) &  \textrm{if} \quad  \beta(Q, C_{P,P'}, \psi) \neq0\\
     0 & \textrm{else}   \end{cases} \\
     \kappa'_{P,P'} &= \begin{cases}\textrm{sgn}(\beta(Q, A_{P,P'}, \psi)) &  \textrm{if} \quad  \beta(Q, A_{P,P'}, \psi) \neq0\\
     0 & \textrm{else}   \end{cases}
\end{align}

We first show that this gives the correct contributions to $\Tr(Q\mathcal{E}[\psi])$ for the $S/C/A$ terms in the first row/column of $M$. The S term has contribution $s_P\kappa_{P,P}^2 = s_P$, as desired. The other terms ($C_{P,P'}$ and $A_{P,P'}$ for $P' \neq P$) have contributions 
\begin{align}
    (c_{P,P'}+ia_{P,P'})(-i\kappa_{P,P'}+\kappa'_{P,P'})(i\kappa_{P,P})+(c_{P,P'}-ia_{P,P'})(i\kappa_{P,P'}+\kappa'_{P,P'})(-i\kappa_{P,P'}) & = c_{P,P'}(2\kappa_{P,P'}\kappa_{P,P}) - a_{P,P'}(2\kappa_{P,P}\kappa'_{P,P'}),
\end{align}
which gives the correct contribution.
It now remains to show that we get the correct contribution for all terms not in the top row/first column---i.e., terms $C_{P',P''}$ and $A_{P',P''}$ for $P' \neq P$ and  $P'' \neq P$. The contributions we get from Eq.~\eqref{eq:matrix_form} are
\begin{multline}
    (c_{P',P''}+ia_{P',P''})(-i\kappa_{P,P'}+\kappa'_{P,P'})(i\kappa_{P,P''}+\kappa'_{P,P''})+(c_{P',P''}-ia_{P',P''})(i\kappa_{P,P'}+\kappa'_{P,P'})(-i\kappa_{P,P''}+\kappa'_{P,P''}) \\ = 2c_{P',P''}(\kappa_{P,P'}\kappa_{P,P''} + \kappa'_{P,P'}\kappa'_{P,P''}) + 2a_{P',P''}(-\kappa'_{P,P'}\kappa_{P,P''}-\kappa_{P,P'}\kappa'_{P,P''}).
\end{multline}
We first consider the $c_{P',P''}$ part. We know that if $C_{P,P'}$ and $C_{P,P''}$ contribute, then $C_{P',P''}$ contributes with sign $\kappa_{P,P'}\kappa_{P,P''}$, and furthermore this implies that $\kappa'_{P,P'}=\kappa'_{P,P''}=0$, thereby giving the correct coefficient of $c_{P',P''}$. If $C_{P,P'}$ and $C_{P',P''}$ contribute, that would imply $C_{P,P''}$ also contributes. Therefore, if $C_{P,P'}$ contributes and $C_{P,P''}$ does not, it must be that $C_{P',P''}$ does not contribute. Furthermore, because $C_{P,P'}$ contributing implies, $A_{P,P'}$ does not contribute,  $\kappa'_{P,P'}\kappa'_{P,P''}=0$, and therefore $\kappa_{P,P'}\kappa_{P,P''} + \kappa'_{P,P'}\kappa'_{P,P''} = 0$, as desired.

We now consider the $a_{P',P''}$ part. There are a few cases, as follows. 

\noindent \textbf{Case 1.} If both $\kappa_{P,P'}$ and $\kappa_{P,P''}$ are nonzero, then $\kappa'_{P,P'}=\kappa'_{P,P''}=0$, and therefore the coefficient of the $a_{P',P''}$ term gives the expected contribution, because $\kappa_{P,P'} \neq 0$ and $\kappa_{P,P''} \neq 0$ imply $C_{P',P''}$ contributes, and therefore $A_{P',P''}$ cannot contribute. 

\noindent \textbf{Case 2.} For the case $\kappa_{P,P'}=\kappa_{P,P''}=0$, we need to show that $A_{P',P''}$ cannot contribute to obtain the correct contribution to $\langle Q \rangle$ from $a_{P',P''}$. Assume that $[P,P']=[P,P'']=0$. We have that $iP'P''\ket{\psi} = iP'PPP''\ket{\psi} = \gamma_{P',P} i P'PPP''\ket{\psi}$, where $\gamma_{P',P}=1$ if $[P',P]=0$ and $\gamma_{P',P}=-1$ otherwise. However, because $\kappa_{P,P'}=\kappa_{P,P''}=0$, $PP'PP''\ket{\psi} \neq \pm \ket{\psi}$. If $PP''\ket{\psi} = \pm i \ket{\psi}$ and $PP'\ket{\psi} = \pm i \ket{\psi}$, then it must be that $[P,P'] \neq 0$ and $[P,P''] \neq 0$, which contradicts the original assumption. 

\noindent \textbf{Case 3.}  We now consider $\kappa_{P,P'}=0$ and $\kappa_{P,P''} \neq 0$. This implies that $[P,P''] \neq 0$ and $iPP'' = \kappa'_{P,P''}\ket{\psi} = \pm \ket{\psi}$. We have $iP'P''\ket{\psi} = iP'PPP''\ket{\psi} = i\kappa_{P,P''}P'P\ket{\psi}$. If $[P,P']=0$, this implies that $PP'\ket{\psi} \neq \pm \ket{\psi}$ (because we assumed that $\kappa_{P,P'}=0$), and therefore $iP'P''\ket{\psi} \neq \pm \ket{\psi}$. Otherwise, $iP'P''\ket{\psi} = -\kappa_{P,P''}\kappa_{P,P'}$, as desired. 

\noindent \textbf{Case 4.}  Finally, we consider $\kappa_{P,P'} \neq 0$ and $\kappa_{P,P''} = 0$, which proceeds analogously to the case above. This implies that $[P,P'] \neq 0$ and $iPP'\ket{\psi} = \kappa'_{P,P'}\ket{\psi} = \pm \ket{\psi}$. We have $iP'P''\ket{\psi} = iP'PPP''\ket{\psi} = -i\kappa'_{P,P'}PP''\ket{\psi}$. If $[P,P'']=0$, this implies that $PP''\ket{\psi} \neq \pm \ket{\psi}$, and therefore $iP'P''\ket{\psi} \neq \pm \ket{\psi}$. Otherwise, $iP'P''\ket{\psi} = -\kappa'_{P,P'}\kappa_{P,P''}$, as desired. 

This covers all cases, and proves that the matrices we constructed give the correct contributions for each term. The CPTP constraint on the circuit error implies that $M$ is a positive semi-definite matrix. We therefore have that  $\boldsymbol{v}^{\dagger} M \boldsymbol{v} \geq 0$ for \emph{any} vector $\boldsymbol{v}$. Therefore $-2 \boldsymbol{x}^{\dagger} M\boldsymbol{x} < 0$, which implies that the leading order estimate for the rate of any DEM event will be nonnegative for a CPTP error model.

At leading order, our DEM's event rates are  $p_D = \frac{1}{2}\sum_{G \in \mathbb{G}_D} \epsilon_G \beta(Q, G, \psi)$, where $Q$ is a representative $Q \in \mathbb{D}$. This strategy reproduces the event polarizations at leading order. To see this note that at leading order in event probabilities, $\langle \prod_{Q \in D} Q \rangle = 1-\sum_{D: Q \in D}\sum_{G \in \mathbb{G}_D} \epsilon_G \beta(\psi, G, P) = 1-\sum_{D: Q \in D} 2p_D$. This expression is exactly the polarization predicted by our DEM at leading order in the DEM event rates. \qed

\subsection{Anticorrelation in Detection Statistics from Amplitude Damping}

\begin{figure}
    \begin{minipage}{0.3\textwidth}
        \centering
        \includegraphics[width=\textwidth]{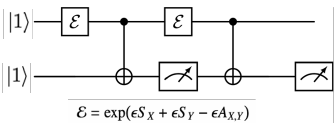}
    \end{minipage}%
    \begin{minipage}{0.7\textwidth}
        \centering
        \text{Detection probabilities}
        \begin{tabular}{|c|c|c|c|}\hline
            Bit String & First-Order Probability & Second-Order Probability & Exact Probability \\\hline
            01 & 0.04 & 0.0353 & 0.0377 \\\hline
            10 & 0.04 & 0.0378 & 0.0393 \\\hline
            11 & 0 & 0.0007 & 0 \\\hline
        \end{tabular}
        \label{tab:detection_probabilities}
        
        \vspace{0.5cm} 
        \text{Detector error models}
        \begin{tabular}{|c|c|c|c|}\hline
            DEM event & First-Order Probability & Second-Order Probability & Exact Probability \\\hline
            01 & 0.04 & 0.0368 & 0.0393 \\\hline
            10 & 0.04 & 0.0392 & 0.0408 \\\hline
            11 & 0 & -0.0016 & -0.0017 \\\hline
        \end{tabular}
        \label{tab:dems_amplitude_damping}
    \end{minipage}
    \caption{\textbf{Anticorrelation in detection statistics.} A single qubit experiencing amplitude damping error and being repeatedly measured indirectly exhibits anticorrelations in its measurements results. In this circuit, only one of the two detectors can flip, never both, because the qubit never returns to the state $\ket{1}$ after decaying to $\ket{0}$. These detection statistics result in a negative-rate DEM event. Using only leading-order expansions, our method accurately predicts the detection statistics at leading order. Beyond leading order, a negative-probability event is needed to compensate for the leading order DEM's prediction of a positive probability for the detection history $11$.}
    \label{fig:anticorrelation}
\end{figure}

Error generator models capture a richer space of dynamics than DEMs, which model error events as stochastic and independent. As a consequence, there are detection statistics that a DEM with nonnegative DEM event probabilities \emph{cannot} capture, which we demonstrate here with a toy example. Consider a single qubit that that experiences an amplitude damping channel
\begin{equation}
\mathcal{E} = \exp(s_x S_X + s_y S_Y - a_{x,y} A_{X,Y}) = \exp(\epsilon S_X + \epsilon S_Y - \epsilon A_{X,Y}),
\end{equation}
where $\epsilon=0.01$, and consider a circuit where this qubit gets measured indirectly twice with a perfect CNOTs to a perfect auxiliary qubit, as shown in Fig.~\ref{fig:anticorrelation}. The detection history probabilities are shown in Fig.~\ref{fig:anticorrelation}. Because the noisy qubit never transitions into the $\ket{1}$ state after transitioning to the $\ket{0}$ state, the first and second detectors  have \emph{anticorrelated} results. As a result, a DEM that predicts the exact probabilities of the detection histories has a negative probability for the event $11$ ($p_{DEM}(11) = -0.0017$). 

In contrast, a first-order approximation to the circuit error (i.e., first-order BCH expansion followed by a first-order Taylor expansion) results in detection probabilities that \emph{can} be modeled by a DEM with nonnegative event rates. Our method gives DEM event probabilities $p_{\{D_1\}}=p_{\{D_2\}}=4\epsilon$ and $p_{\{D_1, D_2\}}=0$ at leading order. This DEM predicts a probability of $16\epsilon^2$ for $11$, which differs from the true probability by $O(\epsilon^2)$. The error in the predicted probabilites of $01$ and $10$ ($p(10) = 0.0393, p(01) = 0.0377$) are also $O(\epsilon^2)$. 

If the approximations in our method are taken beyond first order, DEM events with negative probabilities arise (as shown in Fig.~\ref{fig:anticorrelation} at second order). 
Performing a second-order expansion gives the following polarizations as a function of the $S_X$, $S_Y$, and $A_{X,Y}$ error rates:
\begin{align}
    \lambda_{10} &= 1 -2s_{x} + 2s_{x}^2  -4s_{x}a_{x,y} +  4s_{x}s_{y} + 4a_{x,y} -4a_{x,y}s_{y} -2s_{y} + 2s_{y}^2\\
    &= 1 - 8\epsilon + 16\epsilon^2\\
    \lambda_{11} &= 1 -4s_{x} + 8s_{x}^2 -16s_{x}a_{x,y} + 16s_{x}s_{y} + 8a_{x,y} -16a_{x,y}s_{y} -4s_{y} + 8s_{y}^2\\
    &= 1 - 16\epsilon + 64\epsilon^2\\
    \lambda_{01} &= 1 -2s_{x} + 2s_{x}^2  -12s_{x}a_{x,y} + 4s_{x}s_{y} + 4a_{x,y} + 16a_{x,y}^2 -12a_{x,y}s_{y}  -2s_{y} + 2s_{y}^2 \\
    &= 1 - 8\epsilon + 48\epsilon^2
\end{align}
Estimating a DEM to $O(\epsilon^2)$ from these polarizations gives
\begin{align}
    p_{10}&= 4\epsilon - 8\epsilon^2\\
    p_{01} &= 4\epsilon - 32\epsilon^2\\
    p_{11} &=  -16\epsilon^2.
\end{align}

\section{DEM Events beyond Leading Order}
\label{app:higher_order_events}
Our DEM creation algorithm, when all approximations are performed to leading order, produces a DEM whose events are exactly those that would appear in a Pauli-twirled version of the error model (but with different probabilities in general). Creating a more accurate DEM requires additional DEM events. In this section, we identify a set of DEM events required to model detection statistics to $O(\epsilon^k)$. 

The effects of errors that arise beyond leading order are most straightforwardly seen by looking at a Taylor expansion of the circuit error $\mathcal{E}_c$ (approximated using a BCH expanssion). At order $k$, this expansion contains \emph{compositions} of up to $k$ EEGs from the (propagated) layer error channels, 
\begin{align}
    \exp(\sum_{G \in \mathbb{G}} \epsilon_G G) & \approx 1 + \sum_{G \in \mathbb{G}} \epsilon_G G + \sum_{(G,G') \in \mathbb{G}} \epsilon_G\epsilon_{G'} GG' + \dots \\
    & = 1 + \sum_{G \in T_k} \epsilon_G G, \label{eqn:taylor}
\end{align}
where $T_k$ denotes the set of EEGs required in the $k$th order Taylor expansion of $\mathcal{E}_c$. The result any composition of $k$ EEGs $\{G_i\}_{i=1}^k$ is a linear combination of EEGs whose indices are given by \emph{products} of their indices (see Ref.~\cite{miller2025simulation} for the complete set of composition rules). 

Eq.~\ref{eqn:taylor} encodes how EEGs in $G_{\textrm{EOC}}$ combine to impact the probabilities of detection histories, but it does not immediately reveal which DEM events are needed to generate these detection histories with the correct probabilities. 
To determine which events are needed to model the detection statistics, we use the Taylor expanded error to compute detector polarizations  $\langle \prod_{P \in D_i} P \rangle$, 
\begin{align}
    \langle \prod_{P \in D_i} P \rangle \approx 1 + \sum_{G} \epsilon_{G \in T_k} \beta(\psi, G, \langle \prod_{P \in D} P \rangle).
\end{align}
Each $\beta(\psi, G, \langle \prod_{P \in D} P \rangle)$ can be computed efficiently using the results in Appendix~\ref{app:theory}. The full set of detector polarizations is sufficient to estimate a DEM \cite{blumekohout2025detector}.
We now derive expressions for all $2^n$ detector polarizations in terms of the DEM event classes of the EEGs in $T_k$, and use these expressions to estimate a DEM, allowing us to deduce the set of DEM events needed to reproduce detection statistics to $O(\epsilon^k)$.

\vspace{8pt}
\noindent \textbf{Theorem 2.} Let $E_1$ the the set of all DEM events $\Delta(P)$, where P is a Pauli appearing in the indices of the EEGs in the propagated layer error channels $\mathcal{E}'_1,\mathcal{E}'_{2},\dots ,\mathcal{E}'_k$ (Eq.~\ref{eqn:prod_of_propagated}), and let $E_k$ be the set of DEM events that can be expressed as the symmetric difference of up to $k$ DEM events in $E_1$ (i.e., the net effect of $k$ events happening). Then at $k$th order, the DEM describing the circuit's detector histories to order $k$ contains only events in $E_k$. 
\vspace{8pt}

Before we prove the theorem, we make a few remarks on it. Theorem~\ref{thm:higher_order_events} tells us a polynomial-sized DEM to estimate in order to model detection statistics to $O(\epsilon^k)$. $E_k$ has size $O(n_{\textrm{EEGs}}^k)$,  where $n_{\textrm{EEGs}}$ is the number of EEGs required in the gate-level EEG model for the circuit. In an $S$-only error model, only the events in $E_1$ appear in the final DEM, regardless of $k$. Non-$S$ errors do \emph{not} change the set of detection histories that might appear at $O(\epsilon^k)$ relative to an $S$ only model with EEGs indexed by the same Paulis, but they \emph{can} cause them to occur with different probability than predicted by stochastic Pauli models---requiring the addition of new DEM events to model the detection statistics. However, Theorem~\ref{thm:higher_order_events} does \emph{not} guarantee that the DEM events all have positive rates. In practice, a negative-rate DEM event arising in a DEM suggests that a different model that can model anticorrelations is needed to accurately model the detection statistics. 
\vspace{8pt}

\textbf{Proof of Theorem.} First, we show that the $k$th order BCH approximation of $\mathcal{E}_c$ gives error generators whose indices are products of up to $2k$ indices from the original EEG decomposition. We proceed by casework, using the rules for commutators of EEGs \cite{miller2025simulation}:
\begin{align*}
\left[\eg{H}_P,\eg{S}_Q\right] &= i \eg{C}_{Q,[Q,P]} \\
 \left[\eg{H}_Q,\eg{C}_{A,B}\right] &= i (\eg{C}_{[A,Q],B}+  \eg{C}_{[Q,B],A}) \\
\left[\eg{H}_{P},\eg{A}_{A,B}\right] &= -i(\eg{A}_{[P,A],B}+\eg{A}_{A,[P,B]}) \\
\left[\eg{S}_P,\eg{S}_Q\right] &= 0 \\
\left[\eg{S}_P,\eg{C}_{A,B}\right] & = -i\eg{A}_{AP,BP}-i\eg{A}_{PB,AP}+\frac{-i}{2}\left(\eg{A}_{\{A,B\}P,P}+\eg{A}_{P,P\{A,B\}}\right) \\
[\eg{S}_P,\eg{A}_{A,B}] &= i\left(\eg{C}_{PA,BP}-\eg{C}_{PB,AP}\right)-\frac{i}{2}\eg{A}_{P,[P,[A,B]]} \\ 
[\eg{C}_{A,B},\eg{C}_{P,Q}] &=-i\left(\eg{A}_{AP,QB}+\eg{A}_{AQ,PB}+\eg{A}_{BP,QA}+\eg{A}_{BQ,PA}\right)-\frac{i}{2}\left(\eg{A}_{[P,\{A,B\}],Q}+\eg{A}_{[Q,\{A,B\}],P}+\eg{A}_{[\{P,Q\},A],B}+\eg{A}_{[\{P,Q\},B],A}\right)+\frac{i}{4}\eg{H}_{[\{A,B\}\{P,Q\}]} \\
\left[\eg{C}_{A,B},\eg{A}_{P,Q}\right]&=i\left(\eg{C}_{AP,QB}-\eg{C}_{AQ,PB}+\eg{C}_{BP,QA}-\eg{C}_{PA,BQ}\right)
+\frac{1}{2}\left(\eg{A}_{[A,[P,Q],B}+\eg{A}_{[B,[P,Q]],A}+i\eg{C}_{[P,\{A,B\}],Q}-\eg{C}_{[Q,\{A,B\}],P}\right) -\frac{1}{4}\eg{H}_{[[P,Q],\{A,B\}]} \\
\left[\eg{A}_{A,B},\eg{A}_{P,Q}\right] &= -i\left(\eg{A}_{QB,AP}+\eg{A}_{PA,BQ}+\eg{A}_{BP,QA}+\eg{A}_{AQ,PB}\right) +\frac{1}{2}\left(\eg{C}_{[B,[P,Q]],A}-\eg{C}_{[A,[P,Q]],B}+\eg{C}_{[P,[A,B]],Q}-\eg{C}_{[Q,[A,B]],P}\right)+\frac{i}{4}\eg{H}_{[[P,Q],[A,B]]} \\
\left[\eg{H}_P,\eg{H}_Q\right] &= -i\eg{H}_{[P,Q]} \\
\end{align*}
We first consider the first four commutators above together. These commutator each result in a linear combination of EEGs with (unsigned) indices from the set $\{P_1, P_2, Q_1, Q_2, P_1Q_1, P_1Q_2, P_2Q_1, P_2Q_2\}$. However, we also note that each two-index Pauli contains one of $\{P_1, P_2, Q_1, Q_2\}$ as an index. Therefore, each of the EEGs must have a DEM event class in $E_1$ or the class corresponding to $\emptyset$ (no event). 

We then consider the next two commutators. These commutations results in EEGs with products of three or four index Paulis in them. Again, we note that every one of these has one index from the set $\{P_1, P_2, Q_1, Q_2\}$, implying that these terms do not lead to an additional DEM event class. 

Finally, we consider the last three commutators. Most of the EEGs in the EEG decompositions for these commutators are covered by the cases above. But we also see the appearance of $H$ terms indexed by $P_1P_2Q_1Q_2$, up to sign. These \emph{can} lead to a new DEM event class. However, the existence of $A$ and $C$ terms with $P_1,P_2,$ and $Q_1,Q_2$ as indices implies the existence of the DEM event classes $\Delta(P_1), \Delta(P_2), \Delta(Q_1), \Delta(Q_2)$. This is perhaps concerning, but we note that the emergent term is $O(\epsilon^2)$, and because it is an $H$ term, it has no leading-order effect, so this term can only affect the DEM at $O(\epsilon^4)$.

Next, we show that given an end-of-circuit error generator $G$, a $k$th order expansion of $\exp(G)$ contains only EEGs whose DEM event classes are encoded by bit strings $e_1 \oplus \cdots \oplus e_k$, where $e_1, \dots, e_k$ are events that occur in the leading-order event, as bit strings. This result follows from the composition rules for EEGs (see \cite{miller2025simulation}). The composition of error generators $G_1$ and $G_2$ with indices $P_1, Q_1$ and $P_2, Q_2$ (for an $H$ or $S$ term, let the second index be equal to the first) is a linear combination of EEGs with (unsigned) indices from the set $\{P_1, P_2, Q_1, Q_2, P_1Q_1, P_1Q_2, P_2Q_1, P_2Q_2\}$. Any such error generator has DEM event class  $ D \in \{\Delta(R) \oplus \Delta(R') \mid R, R' \in \{P_1,P_2,Q_1,Q_2\}\}$. We have that $\Delta(P_1),\Delta(P_2),\Delta(Q_1),\Delta(Q_2) \in E_1$. We now need to show that this fact limits the set of possible DEM events. We start by proving a technical lemma about detector polarizations.

\begin{lemma}\label{lem:detector_products}
    Let $G$ be an EEG with DEM event class $D$. Let $P'$ be a detector satisfying $P \notin D$.
    Then $\beta(\psi, G, P'\prod_{P \in D}P) = \beta(\psi, G, \prod_{P \in D}P)$
\end{lemma}

\textbf{Proof of Lemma~\ref{lem:detector_products}:} We consider each type of EEG separately. If $G=H_Q$, then $\beta(\psi, G, P'\prod_{P \in D}P) = 0$ because $P'\prod_{P \in D}P$ is a detector. If $G = S_Q$, then $P' \notin D$ implies that $[P',Q] = 0$ (because $P'$ being a detector implies that $P'\ket{\psi}=\ket{\psi}$). Therefore,  \begin{equation}
    [P'\prod_{P \in D}P,Q] = [P',Q] \prod_{P \in D}P+ P’[\prod_{P \in D}P,Q] = P’[\prod_{P \in D}P,Q],
\end{equation}
which is 0 iff $[\prod_{P \in D}P,Q]=0$. Therefore, $\beta(\psi, G, P'\prod_{P \in D}P) = \beta(\psi, G, \prod_{P \in D}P)$.

We now consider $G=C_{Q_1, Q_2}$. The case of $G=A_{Q_1, Q_2}$ proceeds analogously. In this case, $P' \notin D$ implies that either (a) at least one of $[P’,Q_1] = 0$ or $[P’,Q_2] = 0$ or (b) $Q_1Q_2\ket{\psi} \neq \pm \ket{\psi}$ or $[Q_1,Q_2]=0$. We consider these cases in turn.

Case (a): It cannot be the case exactly one of $[P',Q_1]$ and $[P',Q_2]$ is nonzero unless $D=\emptyset$. But in this case, $\beta(\psi, G, P'\prod_{P \in D}P) = \beta(\psi, G, \prod_{P \in D}P)=0$. Therefore, we are left with considering the case where $[P',Q_1]=[P',Q_2]=0$. By the same reasoning as in the $S$-type EEG case above, $[P'\prod_{P \in D}P,Q_1] = P’[\prod_{P \in D}P,Q_1]$ and $[P'\prod_{P \in D}P,Q_2] = P’[\prod_{P \in D}P,Q_2]$. Therefore, $\beta(\psi, G, P'\prod_{P \in D}P) = \beta(\psi, G, \prod_{P \in D}P)$. 

Case (b): These conditions are independent of $P'$, and they imply $\beta(\psi, G, P'\prod_{P \in D}P)= \beta(\psi, G, \prod_{P \in D}P) = 0$.  \qed

To finish proving the theorem, we prove a final lemma, which will imply that the $O(\epsilon^k)$ approximation of the circuit error can be modeled by a DEM with only the events in $E_k$.

\begin{lemma}\label{lem:dem_event_set}
    Let $\mathbb{G}$ be a set of EEGs, and let $E = \{\eegev(D) \mid G \in \mathbb{G}\}$ be their corresponding DEM events. Then for any $D' \notin E$ the DEM for $1+\sum\epsilon_G G$ does not contain $D'$. Stated differently, the event set $E$ is sufficient for a DEM to model the error $1+\sum\epsilon_G G$ (allowing for the possibility of negative probabilities).
\end{lemma}
\emph{Proof} We need to compute $\langle P \rangle$ for every product of detector Paulis $P$---i.e., we compute all polarizations needed to describe a DEM with detector set $\mathbb{D}$: 
\begin{align}
    \langle \prod_{D} P \rangle & \approx 1 + \sum_{D' \in E}\sum_{G \in \mathbb{G}_{D'}} \beta(\psi, G, \prod_{D} P) \epsilon_{G} \label{eq:pol1} \\
    & = 1 + \sum_{D' \in E}\sum_{G \in \mathbb{G}_i} \beta(\psi, G, \prod_{D \cap D'} P) \epsilon_{G}, \label{eq:pol2}
\end{align}
where we have used Lemma~\ref{lem:detector_products} to go from Eq.~\eqref{eq:pol1} to Eq.~\eqref{eq:pol2}.
Eq.~\eqref{eq:pol2} implies that if $D' \subseteq D$, then $\langle P \rangle \neq \langle P' \rangle$ iff there exists some term that impacts a detector in $D \setminus D'$. 
Furthermore, note that if $|D \cap D'|$ is even, then $\beta(\psi, G, \prod_{D \cap D'} P)=0$ for all $G \in \mathbb{G}_{D'}$, and therefore $\zeta(D',D' \cap D)=0$. 

We now use these polarizations to construct a DEM using the method for dense DEM estimation from \cite{blumekohout2025detector}.
Using our approximate polarizations, the approximate depolarizations are 
\begin{align}
    \omega_D & = -\ln{\lambda_D} \approx \sum_{D' \in E}\sum_{G \in \mathbb{G}_{D'}} \beta(\psi, G, \prod_{D \cap D'} P) \epsilon_{G},
    = \sum_{D' \in E} \zeta(D', D' \cap D),
\end{align}
where we define
\begin{equation}
    \zeta(D', D' \cap D) = \sum_{G \in \mathbb{G}_{D'}} \beta(\psi, G, \prod_{D \cap D'} P) \epsilon_{G}. 
\end{equation}
The resulting event attenuations are given by $\boldsymbol{a} = (-\frac{2}{2^N})H\boldsymbol{\omega}$, where $\boldsymbol{\omega}$ is the vector of all depolarizations. For a $\abs{\mathbb{D}}$-bit string $b$, let $D_b = \{D_i \in \mathbb{D} \mid b_i=1\}$ (where we have picked an ordering of the elements of $\mathbb{D}$). The attenuations can then be expressed as 
\begin{align}
    a_{D_y} & = \left(-\frac{2}{2^N}\right) \sum_{b \in \mathbb{Z}_2^{\abs{\mathbb{D}}}} (-1)^{|D_b \cap D_y|} \omega_b \\
    & = \left(-\frac{2}{2^N}\right) \sum_{\substack{D_b : b \in \mathbb{Z}_2^{\abs{\mathbb{D}}}}} (-1)^{|D_b \cap D_y|} \sum_{D' \in E} \zeta(D', D' \cap D_b) \\
    & = \left(-\frac{2}{2^N}\right) \sum_{D' \in E} \sum_{D_b \subseteq D'} (-1)^{|D_b \cap D_y|} \zeta(D', D_b) \sum_{D_{b'} \subseteq \mathbb{D} \setminus D'} (-1)^{|D_{b'} \cap D_y|}.
\end{align}
Because 
\begin{equation}
    \sum_{D_{b'} \subseteq \mathbb{D} \setminus D'} (-1)^{|D_{b'} \cap D_y|} = 
\begin{cases}
    1 & (\mathbb{D} \setminus D') \cap D_y = \emptyset \\
    0 & \textrm{otherwise}
\end{cases}, 
\end{equation}
it follows that 
\begin{align}
    a_{D_y} & = \left(-\frac{2}{2^N}\right) \sum_{\substack{D' \in E \\ D_y \subseteq D'}} \sum_{D_b \subseteq D'} (-1)^{|D_b \cap D_y|} \zeta(D',D_b). \label{eq:attens_result} 
\end{align}
Eq.~\eqref{eq:attens_result} implies if that there is no $D' \in E$ such that $D_y \subseteq D'$, then $a_{D_y}=0$. It remains to  show that if there is some $D' \in E$ such that $D_y \subseteq D'$, but $D_y \notin E$, then $a_{D_y}=0$. Using the fact that $\zeta(D',D_b)=0$ if $|D_b|$ is even,
\begin{align}
    a_{D_y} & = \left(-\frac{2}{2^N}\right) \sum_{\substack{D' \in E \\ D_y \subseteq D'}} \sum_{\substack{D_b \subseteq D' \\ |D_b|\equiv 1 \textrm{ (mod 2)}}} (-1)^{|D_b \cap D_y|} \zeta(D',D_b). \label{eq:aDy_intermediate}
\end{align}
Suppose that $D' \setminus D_y$ is nonempty. A set of detectors $D_b \subseteq D'$ with odd cardinality can be formed in multiple ways, which we now count. First, $D_b$ can contain an odd number of the elements of $D_y$ and an even number of the elements of $D' \setminus D_y$. The total number of ways this can be done is
\begin{equation}
    \sum_{i=0}^{\lfloor\nicefrac{|D'\setminus D_y|}{2}\rfloor} {|D' \setminus D_y| \choose 2i}\sum_{j=0}^{\lfloor\nicefrac{|D_y|}{2}\rfloor}{|D_y| \choose 2j+1}
\end{equation}
Alternatively, an even number of elements of $D_y$ and an odd number of the elements of $D' \setminus D_y$, and the total number of ways this can be done is
\begin{equation}
    \sum_{i=0}^{\lfloor\nicefrac{|D'\setminus D_y|}{2}\rfloor} {|D' \setminus D_y| \choose 2i+1}\sum_{j=0}^{\lfloor\nicefrac{|D_y|}{2}\rfloor}{|D_y| \choose 2j}.
\end{equation}
Applying the fact that when $D' \setminus D_y \neq \emptyset$ and $D_y \neq \emptyset$,
\begin{equation}
    \sum_{i=0}^{\lfloor\nicefrac{|D'\setminus D_y|}{2}\rfloor} {|D' \setminus D_y| \choose 2i+1}\sum_{j=0}^{\lfloor\nicefrac{|D_y|}{2}\rfloor}{|D_y| \choose 2j} = \sum_{i=0}^{\lfloor\nicefrac{|D'\setminus D_y|}{2}\rfloor} {|D' \setminus D_y| \choose 2i}\sum_{j=0}^{\lfloor\nicefrac{|D_y|}{2}\rfloor}{|D_y| \choose 2j+1},
\end{equation}
to Eq.~\eqref{eq:aDy_intermediate}, we conclude that $a_{D_y}=0$ unless  $D' \setminus D_y = \emptyset$ or $D_y = \emptyset$. This implies that the only nonzero attenuations are $a_{D_y}$ for $D_y \in E$. Because the probability of $D_y$ is $p_{D_y} = \frac{1}{2}(1-\ln a_{D_y})$, the only nonzero probabilities in the DEM correspond to $D_y \in E$. 
\qed

Lemma~\ref{lem:dem_event_set}, applied to a $k$th order expansion of the circuit error,  implies the theorem. \qed

\section{Simulation Details}\label{app:sims}
In this appendix, we provide additional details on our simulations.

\subsection{Statevector Simulation with Coherent Error Models}

We simulated $2$ rounds of surface code syndrome extraction with random coherent error models, comparing our DEM creation method to exact statevector simulations run using  Qiskit ~\cite{qiskit2024} [Fig~\ref{fig:fig1_new}(b)]. We show additional plots from these simulations in Fig.~\ref{fig:h_error_sims}.  We generated 40 error models with random chosen local (crosstalk-free) Hamiltonian errors for each gate (CNOT and H), selecting the error rates so that the two-qubit gate infidelity was approximately $p$, and single-qubit gate infidelity was approximately $0.1p$ for each model, and we varied $p$ from $0.001$ to $0.005$. We generated these error models to be sparse in the EEG basis. For each model, we picked $5$ uniform random two-qubit $H$-type error generators for the CNOT gate's error, and we samples their error rates $h_i$ such that $\sum_{i} h_i^2 = p$---i.e., we sample rates such that all the CNOT has a fixed generator infidelity $p$. For the single-qubit gates, we sample random rates for $h_x, h_y, h_z$ such that $h_x^2+h_y^2+h_z^2=0.1p$. 

\begin{figure}
    \centering
    \includegraphics{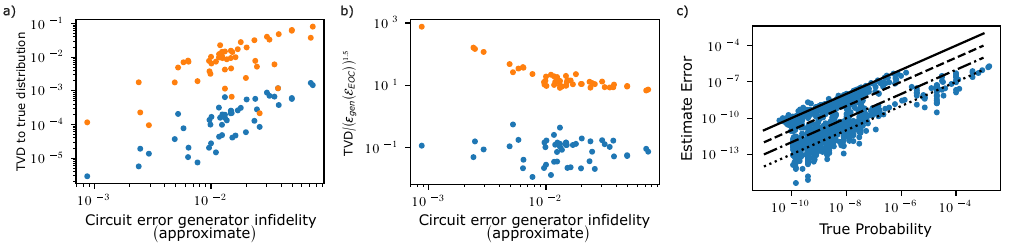}
    \caption{\textbf{Surface codes with random coherent error models.} We created DEMs for $2$ rounds of $d=3$ surface code syndrome extraction with $40$ random crosstalk-free coherent error models. We use these DEMs to compute the predicted probability distribution over detector bit strings, and compare this distribution to the exact distribution, computed via statevector simulation a) Total variation distance (TVD) between the DEM-predicted and true probability distributions for detector bit strings, and the TVD between the distribution estimated from the Pauli-twirled error model and the true distribution. The DEM derived directly from the coherent error model attains a TVD that is typically around $100$ times smaller than the Pauli-twirled model. b) The TVD scales with $\epsilon_{gen}^{1.5}$, where $\epsilon_{gen}$ is the generator infidelity of the circuit's error, indicating that the DEM is accurate to $O(\epsilon^2)$. c) An example of the error in the estimate of one of the $40$ probability distributions produced by a DEM generated with our method. Our method typically estimates the probability of high-probability detector bit strings to within about $1\%$ relative to the true probability, but has varied results on low-probability bit strings.}
    \label{fig:h_error_sims}
\end{figure}

\subsection{Steane Code with General CPTP Error}

We tested our method's ability to model detection statistics for $2-3$ rounds of bare ancilla syndrome extraction for the $[[7,1,3]]$ Steane code with random CPTP error models. We sample from the true probability distribution for detector statistics using the Logical Qubit Simulator (LoQS) \cite{loqs_presentation}. In our simulations of $2$ rounds of syndrome extraction, we generated $200$ random sparse CPTP error models, and we sample $1000$ detection histories for each model using LoQS. We include $12$ detectors for the syndrome measurements and one for the terminating logical $Z$ measurement. To generate each random sparse CPTP error model for our simulations, we do the following. To generate the error model for the two-qubit gates (CNOT), we picked $5$ uniform random two-qubit $H$-type error generators and 5 uniform random $S$-type EEGs, and we sample error rates $h_i$ such that $\sum_{i} h_i^2 = f_2$---i.e., we sample rates such that all the CNOT has a fixed generator infidelity $f_2$. For the single-qubit gates, we sample random rates for $h_x, h_y, h_z$ such that $h_x^2+h_y^2+h_z^2=f_1$

The number of samples we generate with LoQS is small do to the time cost of simulation. As a result, we see most detection histories $<1$ time in expectation in our samples, and therefore we cannot reliably estimate the probability distribution over detection histories from these samples alone. Instead, we use hypothesis testing with a loglikelihood ratio test statistic to check if our DEMs model the sampled data. Specifically, we do a one-sided 5\% significance test to test if there is statistical evidence that the LoQS samples were drawn from a different distribution than the probability distribution described by the DEM. We can perform this loglikelihood test for any coarse-graining of the probability distribution. Given a coarse-graining of the DEM probability distribution $q_{\textrm{DEM}}$, we compute 
\begin{align}
    \mathcal{L}_0 & = \sum n_x \ln(q_x) \\
    \mathcal{L}' & = \sum n_x \ln(n_x/N) \\
    LLR & = -2(\mathcal{L}_0-\mathcal{L}'),
\end{align}
where $n_x$ is the number of counts of observing detector bit string $x$, and $N=\sum_x n_x$ is the total number of counts.

In our Steane code simulations, the DEMs are small enough that we can compute the exact probability distribution over detection histories that the DEM produces. We perform our LLR tests for the full probability distribution as well as one coarse-grained distribution: the probability distribution over all detection bit strings corresponding to a single fault in the standard depolarizing circuit-level noise model for the circuit, with all other outcomes combined together (corresponding to a single bin for ``high-weight'' events). We use a bootstrap to approximate the distribution of the LLR under the assumption that the DEM is the true model for the data, and we use this bootstrapped distribution to determine the $95\%$ threshold of the LLR for our tests. 

Of 200 random models, 190 models (95\% of the models) give a score below the $95$th percentile range of the bootstrapped distribution for the LLR of the entire distribution, and 191 models (95.5\% of the models) lie below the $95$th percentile for the distribution based upon the DEM events caused by low-weight events, showing that our DEMs accurately predict the results of these small-scale QEC circuits. For 50 additional random models, we repeated this test with 5000 shots. We found that 48 of the models (96\%) were below the $95$th percentile for the LLR statistic with the full probability distribution, and 47 of the models (94\%) were were below the $95$th percentile for the LLR statistic with the coarse-grained distribution. 

\subsection{Pauli-twirled error models}
In multiple simulations, we generate and simulate Pauli-twirled versions of the more general error models we simulate. To generate a Pauli-twirled model, we compute the process matrix representation of the post-gate error for each gate, and remove all non-diagonal components to get a Pauli channel. We use a Walsh-Hadamard transform to get the Pauli error rates for the channel, then use those Pauli error rates to generate the detector error model for the circuit of interest with the Pauli-twirled error. Note that we do not attempt to express the Pauli error channel using an error generator. In fact, this procedure can result in Pauli channels that are not CP-divisible, and therefore not representable as an exponentiated error generator. We use Stim to generate the DEM for the Pauli model, to bypass the need for an EEG representation. This choice also has a side effect---each error in the Pauli error channel is treated as independent in Stim when generating the DEM.

\subsection{Gross Code Memory Simulations}
\begin{figure}
    \centering
    \includegraphics{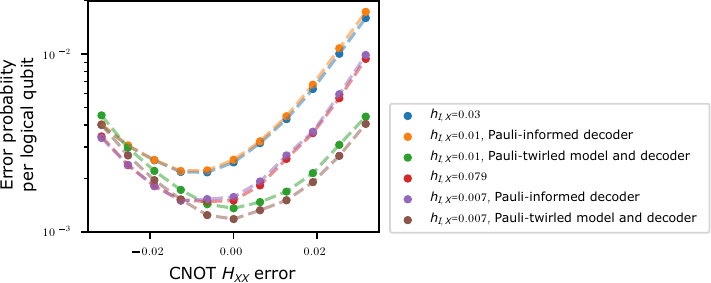}
    \caption{\textbf{Gross code logical error probability with coherent errors and different decoder priors.} Logical error probabilities for varying $H_{XX}$ error on CNOT gates for two fixed strengths of $H_X$ error on idling qubits. In the regime where errors cancel, the Pauli-twirled model overestimates the logical error rate. We also compared decoding with the beam decoder using two different error priors---one from the DEM generated by our method, and the other from the DEM generated from the Pauli twirled model. We find no significant difference between the performance of these two decoder priors.}
    \label{fig:app_bb}
\end{figure}

Here we provide additional details of our simulations of syndrome extraction for the gross code presented in Figure~\ref{fig:bb}. In Figure~\ref{fig:app_bb}, we show the logical error probability per logical qubit for two strengths of $H_X$ error on idle qubits during measurement. We see that the logical error probability is minimized at nonzero CNOT $H_{XX}$ error, as this leads to coherent error cancellation. The minimum depends on the strength of the coherent idle error. In contrast, the Pauli twirled error model always predicts that logical error probability is minimized when the CNOT infidelity is minimized (i.e., when there is no $H_{XX}$ error). 

We also compare logical error probabilities when weighting the beam decoder's priors with the DEM generated by our method to weighting the decoder's priors with the DEM generated from the Pauli twirled model. We see no significant difference in performance between these two decoders. 

\subsection{Simulation Error Model Details}
\subsubsection{Surface code error model}
To generate random error models for our surface code simulations [Fig.~\ref{fig:surface_thresholds}(a)], we selected $10$ random S-type EEGs, $3$ random H-type EEGs, 2 random A-type EEGs, and 2 random C-type EEGs from the 2-qubit EEGs to model the CZ error, and sampled random values for their rates. For $S$, $C$, and $A$ terms, we do this by sampling a random matrix whose elements have mean $6 \times 10^{-5}$ and standard deviation $1.3 \times 10^{-5}$, and use it to construct a positive semidefinite tensor of error rates (using pyGSTi \cite{Nielsen2020-rd}. We sample $H$ EEG rates from a normal distribution with mean $0$ and standard deviation $0.006$. We model $H$ gate error as being single-qubit depolarizing error with infidelity $6.6\times 10^{-6}$, and include pre-measurement/post-reset bit-flip error of $0.001$.

In Table~\ref{tab:surface}, we present the parameters of the error models used in Fig.~\ref{fig:surface_thresholds}(b).

\begin{table}[h]
    \centering
    \caption{\textbf{Error models for surface code simulations.} We simulated rotated surface code syndrome extraction with error model families defined by a parameter $p$, where each family had a different amount of coherent error determined by a parameter $h$. In our simulations, we varied $h$ between $0$ and $1$, and we varied $p$ to identify the threshold for the corresponding single-parameter error model family.}\label{tab:surface}
    \begin{tabular}{|c|c|c|}
        \hline
        Gate Type & Error Type & Error Rate \\ \hline
        H & (H, Z) & \(\sqrt{0.0001p \cdot h}\) \\ \hline
        H & (H, X) & \(\sqrt{0.0001p \cdot h}\) \\ \hline
        CNOT & (H, IX) & \(\sqrt{0.01p \cdot h}\) \\ \hline
        CNOT & (H, ZI) & \(\sqrt{0.01p \cdot h}\) \\ \hline
        CNOT & (H, ZX) & \(\sqrt{0.01p \cdot h}\) \\ \hline
        H & (S, Z) & \(0.0001p + 0.0001p \cdot (1 - h)\) \\ \hline
        H & (S, Y) & \(0.0001p\) \\ \hline
        H & (S, X) & \(0.0001p \cdot (1 - h)\) \\ \hline
        CNOT & (S, IX) & \(\frac{0.01p}{15} + 0.01p \cdot (1 - h)\) \\ \hline
        CNOT & (S, IY) & \(\frac{0.01p}{15}\) \\ \hline
        CNOT & (S, IZ) & \(\frac{0.01p}{15}\) \\ \hline
        CNOT & (S, XX) & \(\frac{0.01p}{15}\) \\ \hline
        CNOT & (S, XY) & \(\frac{0.01p}{15}\) \\ \hline
        CNOT & (S, XZ) & \(\frac{0.01p}{15}\) \\ \hline
        CNOT & (S, YX) & \(\frac{0.01p}{15}\) \\ \hline
        CNOT & (S, YY) & \(\frac{0.01p}{15}\) \\ \hline
        CNOT & (S, YZ) & \(\frac{0.01p}{15}\) \\ \hline
        CNOT & (S, ZX) & \(\frac{0.01p}{15} + 0.01p \cdot (1 - h)\) \\ \hline
        CNOT & (S, ZY) & \(\frac{0.01p}{15}\) \\ \hline
        CNOT & (S, ZZ) & \(\frac{0.01p}{15}\) \\ \hline
        CNOT & (S, XI) & \(\frac{0.01p}{15}\) \\ \hline
        CNOT & (S, YI) & \(\frac{0.01p}{15}\) \\ \hline
        CNOT & (S, ZI) & \(\frac{0.01p}{15} + 0.01p \cdot (1 - h)\) \\ \hline
        Pre-measurement error & (S, X) & 0 \\ \hline
        State preparation/post-measurement error & (S, X) & 0 \\ \hline
    \end{tabular}
\end{table}
 
\subsubsection{Gross code error model}
In Table~\ref{tab:bb}, we present the parameters of the error models used in Fig.~\ref{fig:bb}.

\begin{table}[h]
    \centering
    \caption{\textbf{Error models for BB code simulations.} We simulated cultivation with varied amounts of coherent error, as parameterized by two parameters, $h_1$ and $h_2$. We used $p=0.8$, varied $h_1$ between $-1$  and $1$, and we varied $h_2$ between $0$ and $1$.}\label{tab:bb}
    \begin{tabular}{|l|l|l|}
        \hline
        Gate Type & Error Type & Error Rate \\ \hline
        I & (S, X) & \(0.0008\) \\ \hline
        I & (S, Y) & \(0.0008\) \\ \hline
        I & (S, Z) & \(0.0008\) \\ \hline
        I & (H, X) & \(\sqrt{0.00024} \cdot h_{2}\) \\ \hline
        CNOT & (H, XX) & \(\sqrt{0.0008} \cdot h_{1}\) \\ \hline
        H & (S, Z) & \(0.0008\) \\ \hline
        H & (S, Y) & \(0.0008\) \\ \hline
        H & (S, X) & \(0.0008\) \\ \hline
        CNOT & (S, IX) & \(0.0008\) \\ \hline
        CNOT & (S, IY) & \(0.0008\) \\ \hline
        CNOT & (S, IZ) & \(0.0008\) \\ \hline
        CNOT & (S, XI) & \(0.0008\) \\ \hline
        CNOT & (S, YI) & \(0.0008\) \\ \hline
        CNOT & (S, ZI) & \(0.0008\) \\ \hline
        Pre-measurement error  & (S, X) & 0.001 \\ \hline
        State preparation/post-measurement error & (S, X) & \(0\) \\ \hline
    \end{tabular}
\end{table}

\subsubsection{Cultivation error model}
In Table~\ref{tab:cultiv}, we present the parameters of the error models used in our simulations of magic state cultivation (with an $S$ gate proxy). Fig.~\ref{fig:surface_thresholds}. We fix a $CZ$ gate generator infidelity of $0.00165$ and single-qubit gate generator infidelities of $0.00042$ and vary the fraction of Z-type error ($ZI$, $IZ$, and $ZZ$ error) on CZ gates and Z error on single-qubit gates gates) that is coherent ($H_{ZI}$,$H_{IZ}$,$H_{ZZ}$) versus Pauli stochastic ($S_{ZI}$,$S_{IZ}$,$S_{ZZ}$).

\begin{table}[h]
    \centering
    \caption{\textbf{Error models for cultivation simulations.} We simulated cultivation with varied amounts of coherent error, as parameterized by a single parameter $h$. We varied the value of $h$ from 0 to 0.35.}\label{tab:cultiv}
    \begin{tabular}{|c|c|c|}
        \hline
        Gate Type & Error Type & Parameter \\ \hline
        Non-idle single-qubit gates & (S, X) & \(0.1 \times 10^{-3} \cdot 0.6\) \\ \hline
        Non-idle single-qubit gates & (S, Z) & \(0.4 \times 10^{-3} \cdot 0.6 \cdot (1 - h) + 0.1 \times 10^{-3} \cdot 0.6\) \\ \hline
        Non-idle single-qubit gates & (S, Y) & \(0.1 \times 10^{-3} \cdot 0.6\) \\ \hline
        Non-idle single-qubit gates & (H, Z) & \(\sqrt{0.4 \times 10^{-3} \cdot h \cdot 0.6}\) \\ \hline
        I & (S, X) & \(0.1 \times 10^{-3} \cdot 0.6\) \\ \hline
        I & (S, Z) & \(0.4 \times 10^{-3} \cdot 0.6 \cdot (1 - h) + 0.1 \times 10^{-3} \cdot 0.6\) \\ \hline
        I & (S, Y) & \(0.1 \times 10^{-3} \cdot 0.6\) \\ \hline
        I & (H, X) & \(0\) \\ \hline
        I & (H, Z) & \(\sqrt{0.4 \times 10^{-3} \cdot h \cdot 0.6}\) \\ \hline
        I & (H, Y) & \(0\) \\ \hline
        CZ & (S, IX) & \(5 \times 10^{-5} \cdot 0.6\) \\ \hline
        CZ & (S, IY) & \(5 \times 10^{-5} \cdot 0.6\) \\ \hline
        CZ & (S, IZ) & \(5 \times 10^{-4} \cdot 0.6 \cdot (1 - h) + 5 \times 10^{-5} \cdot 0.6\) \\ \hline
        CZ & (S, XI) & \(5 \times 10^{-5} \cdot 0.6\) \\ \hline
        CZ & (S, YI) & \(5 \times 10^{-5} \cdot 0.6\) \\ \hline
        CZ & (S, ZI) & \(5 \times 10^{-4} \cdot 0.6 \cdot (1 - h) + 5 \times 10^{-5} \cdot 0.6\) \\ \hline
        CZ & (S, XX) & \(5 \times 10^{-5} \cdot 0.6\) \\ \hline
        CZ & (S, XY) & \(5 \times 10^{-5} \cdot 0.6\) \\ \hline
        CZ & (S, XZ) & \(5 \times 10^{-5} \cdot 0.6\) \\ \hline
        CZ & (S, YX) & \(5 \times 10^{-5} \cdot 0.6\) \\ \hline
        CZ & (S, YY) & \(5 \times 10^{-5} \cdot 0.6\) \\ \hline
        CZ & (S, YZ) & \(5 \times 10^{-5} \cdot 0.6\) \\ \hline
        CZ & (S, ZX) & \(5 \times 10^{-5} \cdot 0.6\) \\ \hline
        CZ & (S, ZY) & \(5 \times 10^{-5} \cdot 0.6\) \\ \hline
        CZ & (S, ZZ) & \(1 \times 10^{-3} \cdot 0.6 \cdot (1 - h) + 5 \times 10^{-5} \cdot 0.6\) \\ \hline
        CZ & (H, ZI) & \(\sqrt{5 \times 10^{-4} \cdot 0.6 \cdot h}\) \\ \hline
        CZ & (H, IZ) & \(\sqrt{5 \times 10^{-4} \cdot 0.6 \cdot h}\) \\ \hline
        CZ & (H, ZZ) & \(\sqrt{(1 \times 10^{-3}) \cdot 0.6 \cdot h}\) \\ \hline
        Pre-measurement error & (S, X) & 0.0006 \\ \hline
        State preparation/post-measurement error & (S, X) & 0.0006 \\ \hline
    \end{tabular}
\end{table}
 %TC:endignore
\end{document}